\documentclass[%
 reprint,
amsmath,amssymb,
aps,pra,twocolumn, superscriptaddress, showpacs
]{revtex4-1}

\def \bra#1{\mathinner{\langle{#1|}}}
\def \ket#1{\mathinner{|{#1}\rangle}}

\def \red #1 {\textcolor{red}{#1}}
\def \blue #1 {\textcolor{blue}{#1}}

\usepackage{amsthm}

\usepackage{amsmath}

\usepackage{tabularx}

\usepackage{ulem}
\usepackage{graphicx}

\usepackage{dcolumn}
\usepackage{bm}
\usepackage{mathtools}
\usepackage{xcolor}
\usepackage{framed}

\usepackage{setspace} 
\usepackage{float}
\usepackage{algorithm}
\usepackage{algpseudocode}
\newcommand\Algphase[1]{%
\vspace*{-.5\baselineskip}\Statex\hspace*{\dimexpr-\algorithmicindent-2pt\relax}\rule{0.98\textwidth}{0.4pt}%
\Statex\hspace*{-\algorithmicindent}\textbf{#1}%
\vspace*{-.5\baselineskip}\Statex\hspace*{\dimexpr-\algorithmicindent-2pt\relax}\rule{0.98\textwidth}{0.4pt}%
}

\usepackage{tikz}
\usetikzlibrary{automata, arrows}
\usepackage{circuitikz}
\usetikzlibrary{backgrounds,fit,decorations.pathreplacing,calc,backgrounds}  
\tikzstyle{operator} = [draw,fill=white,minimum size=1.5em] 
\tikzstyle{operator2}=[draw,fill=white, text width=0.6cm, minimum height=4cm] 
\tikzstyle{operator3}=[draw,fill=white, text width=0.8cm, minimum height=1cm] 
\tikzstyle{phase} = [draw, fill=white,shape=circle,minimum size=12pt,inner sep=0pt]
\tikzstyle{phase1} = [fill=black, shape=circle,minimum size=3.5pt,inner sep=0pt]
\tikzstyle{phase2} = [fill=blue, shape=circle,minimum size=7pt,inner sep=1pt]
\tikzstyle{phase11} = [fill=cyan, shape=circle,minimum size=7pt,inner sep=1pt]
\tikzstyle{phase12} = [fill=lime, shape=circle,minimum size=7pt,inner sep=1pt]
\tikzstyle{phase13} = [fill=pink, shape=circle,minimum size=7pt,inner sep=1pt]
\tikzstyle{phase14} = [fill=lightgray, shape=circle,minimum size=7pt,inner sep=1pt]
\tikzstyle{phase0} = [draw, fill=white, shape=circle,minimum size=3.5pt,inner sep=0pt]
\tikzstyle{ellipsis} = [fill,shape=circle,minimum size=2pt,inner sep=0pt]
\tikzset{meter/.append style={fill=white, draw, inner sep=3, rectangle, font=\vphantom{A}, minimum width=14, 
path picture={\draw[black] ([shift={(.05,.2)}]path picture bounding box.south west) to[bend left=40] ([shift={(-.05,.2)}]path picture bounding box.south east);\draw[black,-latex] ([shift={(0,.15)}]path picture bounding box.south) -- ([shift={(.15,-.08)}]path picture bounding box.north);}}}
\tikzset{cross/.style={path picture={\draw[thick,black](path picture bounding box.north) -- (path picture bounding box.south) (path picture bounding box.west) -- (path picture bounding box.east);
}},
crossx/.style={path picture={\draw[thick,black,inner sep=0pt]
(path picture bounding box.south east) -- (path picture bounding box.north west) (path picture bounding box.south west) -- (path picture bounding box.north east);
}},
circlewc/.style={draw,circle,cross,minimum width=0.3 cm},
}
\usepackage{hyperref}
\hypersetup{
    colorlinks=true,
    linkcolor=blue,
    citecolor=blue,
    filecolor=blue,
    urlcolor=blue
}

\makeatletter
\renewcommand{\boxed}[2][\fboxsep]{{%
  \setlength{\fboxsep}{#1\fboxsep}\fbox{\m@th$\displaystyle#2$}}}
\makeatother

\usepackage{tikz} 
\usetikzlibrary{backgrounds,fit,decorations.pathreplacing,calc,backgrounds,calligraphy}  

\newtheorem{theorem}{Theorem}
\newtheorem{lemma}{Lemma}[theorem]
\newtheorem{res}{Result}

\newtheorem{proposition}{Proposition}

\begin{document}


\title{Learning Parameterized Quantum Circuits with Quantum Gradient}

\author{Keren Li}
\email{likr@szu.edu.cn}
\affiliation{College of Physics and Optoelectronic Engineering, Shenzhen University, Shenzhen 518060, China}

\author{Yuanfeng Wang}
\affiliation{Quantum Science Center of Guangdong-Hong Kong-Macao Greater Bay Area (Guangdong), Shenzhen 518045, China}

\author{Pan Gao}
\email{gaopan@baqis.ac.cn}
\affiliation{Beijing Academy of Quantum Information Sciences, Beijing 100193, China}

\author{Shenggen Zheng}
\email{zhengshenggen@quantumsc.cn}
\affiliation{Quantum Science Center of Guangdong-Hong Kong-Macao Greater Bay Area (Guangdong), Shenzhen 518045, China}

\date{\today}

\begin{abstract}
Parameterized quantum circuits (PQCs) are crucial for quantum machine learning and circuit synthesis, enabling the practical implementation of complex quantum tasks. 
However, PQC learning has been largely confined to classical optimization methods, which suffer from issues like gradient vanishing.
In this work, we introduce a nested optimization model that leverages quantum gradient to enhance PQC learning for polynomial-type cost functions.
Our approach utilizes quantum algorithms to identify and overcome a type of gradient vanishing—a persistent challenge in PQC learning—by effectively navigating the optimization landscape.
We also mitigate potential barren plateaus of our model and manage the learning cost via restricting the optimization region. Numerically, we demonstrate the feasibility of the approach on two tasks: the Max-Cut problem and polynomial optimization. The method excels in generating circuits without gradient vanishing and effectively optimizes the cost function.
From the perspective of quantum algorithms, our model improves quantum optimization for polynomial-type cost functions, addressing the challenge of exponential sample complexity growth. 
\end{abstract}

\maketitle


\section{Introduction}
\label{Intro}
Quantum circuit synthesis enables a wide range of applications, including unitary approximation, state preparation, ansatz design, and non-linear simulation~\cite{furrutter2024quantum}. 
Central to this process is the learning of \textit{parameterized quantum circuits} (PQCs), which incorporate adjustable parameters into quantum gates. By fine-tuning these parameters, PQCs can represent complex quantum states and operations, making them essential tools in quantum computation and quantum machine learning~\cite{benedetti2019parameterized}.
For the noisy intermediate-scale quantum devices, PQCs are significant due to their adaptability and potential for utilizing limited quantum resources. They are invaluable for mitigating noise and facilitating hybrid quantum-classical approaches~\cite{cerezo2021variational, bharti2022noisy}. For the fault-tolerant quantum computers, PQCs are essential to provide a universal, scalable, and adaptable framework for constructing subroutines of quantum protocols. For example, PQCs are fundamental in quantum machine learning, enabling a streamlined representation of quantum neural networks architectures~\cite{abbas2021power, beer2020training}.

Finding an appropriate PQC is crucial—a task commonly known as circuit synthesis—which involves exploring both the circuit architecture and its parameters. 
However, learning PQCs using only classical optimization methods faces significant challenges. A typical way is to pre-define an ansatz and train it using a classical optimizer, with well-established techniques such as the parameter shift rule for gradient-based optimization~\cite{mitarai2018quantum}. As the size of the problem increases, pre-definition and train within this high-dimensional Hilbert space becomes computationally challenging. 
Moreover, the persistent issue of gradient vanishing, caused not only by barren plateaus but just by the exponentially expanding Hilbert space, remains a major obstacle in fully leveraging the potential of PQCs. Various protocols have been proposed to mitigate barren plateaus~\cite{wang2021noise, mcclean2018barren}, including the use of shallow or variable structure ansatz~\cite{cerezo2021cost, grimsley2023adaptive}, specialized initialization strategies~\cite{zhang2022escaping}, and tailored architectures~\cite{pesah2021absence}. However, methods to address gradient vanishing caused just by the exponentially expanding Hilbert space are scarce.

For the limitations of classical methods in overcoming these obstacles, naturally, it raises a question: \textit{How about leveraging quantum algorithms to learn PQCs?}

Given that research in this area is quite limited, in this work, we propose a novel approach that leverages quantum algorithms to enhance PQC learning. Specifically, we introduce a \textit{nested optimization model} (NOM) for polynomial-type cost functions, utilizing quantum gradient algorithms to optimize directly within the geometry of the Hilbert space. By extending quantum gradient algorithms~\cite{rebentrost2019quantum, li2021optimizing, li2020quantum, gao2021quantum} from the real domain to the complex domain, our model effectively explores the optimization landscape, identifying and overcoming the persistent gradient vanishing issues encountered in Section \ref{vanishing}.

Our main contributions are as follows:
First, we extend the quantum gradient algorithm from the real domain to the complex domain, creating a NOM for PQC learning with quantum resources.
Second, we address gradient vanishing in PQC learning caused by inadequate or suboptimal parameter configurations, thereby enhancing training efficiency.
Third, by managing the optimization region, we mitigate potential barren plateaus in the classical learning part of NOM. Additionally, by integrating a reinforcement learning (RL)-based method, we numerically demonstrate the effectiveness of the protocol.
Finally, the NOM can address the key issue related to sample efficiency in quantum gradient algorithm, reducing the overall resource overhead.

\section{Results}
Quantum circuit synthesis is the process of generating a quantum circuit that meets a specific  target. In this work, we consider the problem of minimizing a polynomial-type function where the variables are encoded in a quantum state generated by a PQC. We formulate the problem as,
\begin{eqnarray}\label{problem_state}
   \min \quad f(\bm{z}) = \bm{z}^{\dagger \otimes p} \mathcal{F} \bm{z}^{\otimes p},
\end{eqnarray}
where $\bm{z}=(1, z_1,\cdots, z_d)^T\in \mathbb{C}^{d+1}$ is with up to $d$ complex variables, and $f(\bm{z}): \mathbb{C}^{d+1} \rightarrow \mathbb{R}$ denotes a polynomial function with $2p$ degrees, characterized by a Hermitian matrix $\mathcal{F}$. 
$\bm{z}$ can be encoded in a quantum state, 
\begin{eqnarray} \label{amp_encoding}
  \ket{\bm{z}} = c_0 \ket{0} + \sum_{i=1}^{d} c_i \ket{i},
\end{eqnarray}
which is generated from a PQC, where $c_0=1/\sqrt{1+\sum_{i=1}^{d}|z_i|^2}$ and $c_i = c_0 \cdot z_i$. 
Minimization of  $f(\bm{z})$ is a polynomial optimization problem, which encompasses common scenarios encountered in both classical and quantum computation. 
When $\bm{z} \in \mathbb{R}^{d+1}$, it corresponds to the real-valued polynomial optimization domain, which has been studied in~\cite{rebentrost2019quantum, li2020quantum}. 
When $p = 1$ and $\mathcal{F}_{1,i} = \mathcal{F}_{j,1} = 0$ for $i, j \in [d+1] $, the problem simplifies to finding the ground state, where $\mathcal{F}$ encodes the information of Hamiltonian~\cite{farhi2014quantum, peruzzo2014variational}.

\subsection{Gradient Vanishing}
\label{vanishing}
A local minimum of Eq.~(\ref{problem_state}) should satisfies such criterion,
\begin{equation} \label{condition_1}
    \|\nabla_{\bm{z}} f\| = \left\|\left(\frac{\partial f}{\partial z_1},\dots,\frac{\partial f}{\partial z_d}\right)^{T}\right\|=0.
\end{equation} 
In traditional method, the state $\ket{\bm{z}}$ is substituted with a $m$-variable ansatz, leading to a modified criterion,
\begin{equation}\label{condition_2}
     \|\nabla_{\bm{\theta}} f\| = \left\|\left(\frac{\partial f}{\partial \theta_1},\dots,\frac{\partial f}{\partial \theta_m}\right)^{T}\right\|=0.
\end{equation} 
A likely scenario is that $\|\nabla_{\bm{z}} f\| \neq 0$ but $\|\nabla_{\bm{\theta}} f\| = 0$. This arises because 
\begin{equation}\label{derivatives}
        \nabla_{\bm{\theta}} f = \frac{\partial \bm{z}}{\partial \bm{\theta}} \cdot \nabla_{\bm{z}} f, \quad \text{where}\quad \frac{\partial{\bm{z}}}{\partial \bm{\theta}} = \begin{pmatrix}
             \frac{\partial z_1}{\partial \theta_1} & \dots & \frac{\partial z_d}{\partial \theta_1} \\
             \vdots & \ddots & \vdots \\
             \frac{\partial z_1}{\partial \theta_m} & \dots & \frac{\partial z_d}{\partial \theta_m} 
    \end{pmatrix},
\end{equation}
indicates that $\nabla_{\bm{\theta}} f = 0$ is a necessary but not sufficient condition for $\nabla_{\bm{z}} f = 0$. In Result \ref{res2}, we evaluate this probability under certain assumptions.

\begin{res} \label{res2}
Given the norm of $\nabla_z f$ and any $i$-th row of $\frac{\partial{\bm{z}}}{\partial \bm{\theta}_i}$ bounded by $C$, which is independent of $d$, the dimension of the Hilbert space, we have
\begin{eqnarray}
    \mbox{Pr}(||\langle \frac{\partial{\bm{z}}}{\partial \bm{\theta}_i}|\nabla_z f\rangle||\geq \epsilon ) \leq 8 e^{-\frac{d\epsilon^2}{128 C^2}},
\end{eqnarray}
which implies that the probability of the inner product $\langle \frac{\partial{\bm{z}}}{\partial \bm{\theta}_i} | \nabla_z f \rangle$ exceeding a certain threshold decays exponentially as $d$ increases.
Furthermore, with the number of parameters $\bm{\theta}$ denoted as $m$, we have 
$||\bm{\nabla_\theta}f||=\sqrt{\sum_{i=1}^{m}||\langle \frac{\partial{\bm{z}}}{\partial \bm{\theta}_i} | \bm{\nabla_z}f\rangle ||^2}$, which leads to
\begin{eqnarray}
    &&\text{Pr}\left(||\bm{\nabla_\theta}f||\geq  \epsilon \right)\leq 8m \cdot  e^{-\frac{d \epsilon^2}{128 C^2 m}}
\end{eqnarray}

\end{res}
\begin{proof}
    They are provided in Appendix \ref{supp:gradient_vanishing}~\cite{supp}. 
\end{proof}

Result \ref{res2} indicates that in variational quantum algorithms, when the number of parameters is significantly smaller than the Hilbert space dimension (commonly, $m$ scales polynomially with $\log(d)$), the gradient with respect to the parameters will decay exponentially with high probability as the Hilbert space dimension increases.
Additionally, the assumptions given here are reasonable in common scenarios: $m \sim \log(d)$ is necessary for constructing an executable PQC, and the limited norms on $\frac{\partial \bm{z}}{\partial \theta}$ and $\nabla_z f$ reflect the reality. As a result, the gradient vanishing issue described earlier poses a significant challenge throughout traditional methods.

\subsection{Nested Optimization Model (NOM)}
\label{theory}

\begin{figure*}[!ht]
   \centering
   \includegraphics[width=1.8\columnwidth]{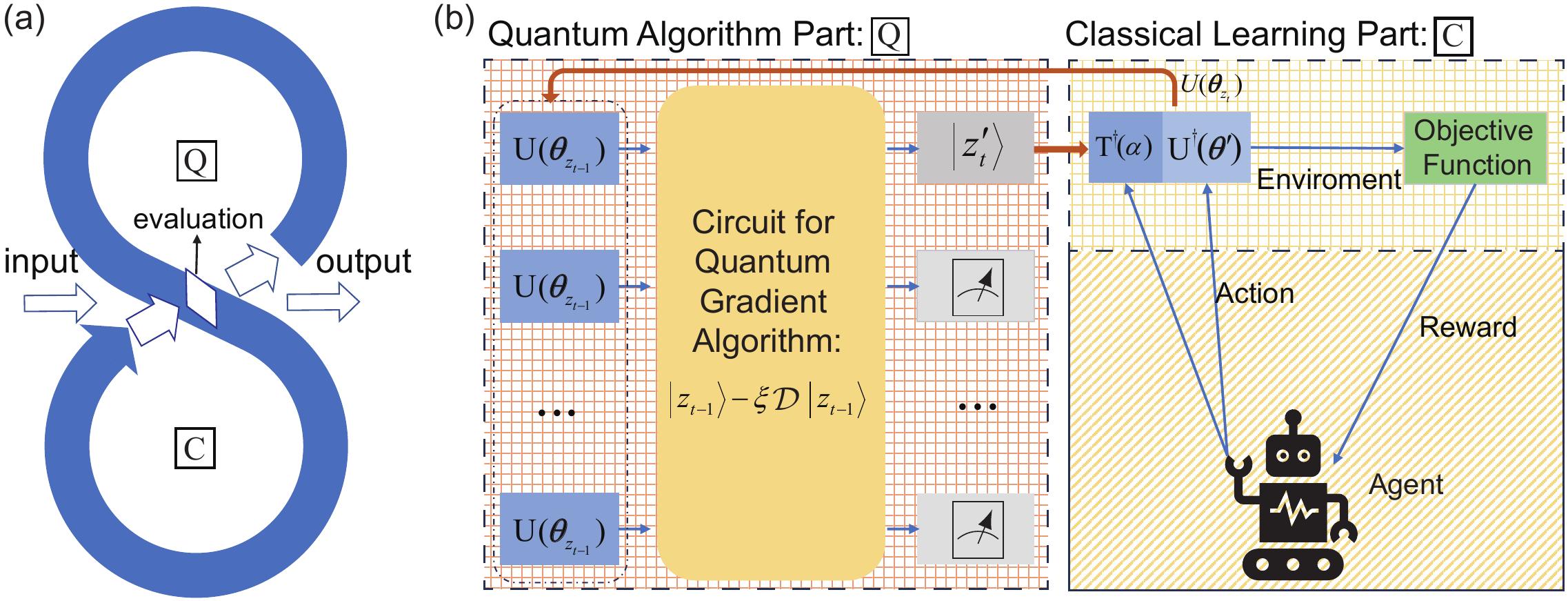}
   \caption{The Flow of NOM. (a) illustrates the protocol flow, where $\boxed[.5]{Q}$ and $\boxed[.5]{C}$ represent the quantum algorithm and classical learning parts, respectively. (b) depicts the process of the $t$-th iteration in the protocol. In this step, $\boxed[.5]{Q}$ passes the state $\ket{\bm{z}'}$ to $\boxed[.5]{C}$, and $\boxed[.5]{C}$ conditionally returns $U(\bm{\theta}_{z_{t}})$ to $\boxed[.5]{Q}$. The dashed arrow indicates a decision point, corresponding to the diamond shape shown in (a). } 
\label{fig1} 
\end{figure*}

To above problem, we introduce a NOM (depicted in Figure \ref{fig1}(a)), which integrates the PQC learning with quantum gradient. The protocol is heuristic, which requires an initial PQC as input, and two subroutines for iterative learning: $\boxed[.5]{Q}$ and $\boxed[.5]{C}$, depicted in Figure \ref{fig1}(b). The label $\boxed[.5]{Q}$ represents a quantum algorithm that performs,
\begin{eqnarray}\label{Q_loop}
    \ket{\bm{z}'_t}=\ket{\bm{z}_{t-1}}-\xi \mathcal{D} \ket{\bm{z}_{t-1}}.
\end{eqnarray}  
Here, $\xi$ is learning rate, $\mathcal{D}$ is the effective gradient operator, and $\mathcal{D}\ket{\bm{z}} \propto \ket{\nabla f(\bm{z})}$, where $\nabla f(\bm{z})$ represents the steepest descent direction for $f(\bm{z})$.
Notably, $\ket{\bm{z}_{t-1}}$, different from $\ket{\bm{z'}_{t-1}}$, can be represented by a PQC, noted as $\ket{\bm{z}_{t-1}}=U(\bm{\theta_{\bm{z}_{t-1}}})\ket{0}$. 
On the LHS of Eq.~\eqref{Q_loop}, $\ket{\bm{z}'_t}$ does not have an explicit PQC representation.
Thus a classical learning part $\boxed[.5]{C}$ is required to obtain the PQC representation, $U(\bm{\theta_{\bm{z}_{t}}})$, via 
\begin{eqnarray}\label{C_loop}
    \text{min}_{\bm{\theta}} 1-\| \bra{\bm{z}'_t} [U(\bm{\theta})]_t \ket{0} \|^2.
\end{eqnarray}
Here we refer to $[U(\bm{\theta})]_t$ as a PQC at $t$-th iteration to be learned with adjustable architecture and parameters, and the resultant state $\ket{\bm{z}_t}=U(\bm{\theta_{\bm{z}_{t}}})\ket{0}$ approximates $\ket{\bm{z}_t'}$. In this work, we will demonstrate a RL protocol working for it, which is shown in Figure \ref{fig1}(b).

We summarize the flow of NOM below.
Step 1: an initial $U(\bm{\theta_{\bm{z}_{0}}})$, encoding the initial value  $\bm{z}_0$, is input into the algorithm. 
Step 2: using pre-set criteria, we evaluate the cost function to determine whether to proceed with further optimization steps (this decision point is represented by the diamond in Figure~\ref{fig1}(a)).
Step 3: if further optimization is required, we apply a quantum gradient algorithm to generate a new quantum state  $\ket{\bm{z}'_1}$ as shown in Eq.~\eqref{Q_loop}.
Step 4: $\ket{\bm{z}'_1}$ is approximated through a classical learning part as $\ket{\bm{z}_1}$, with $U(\bm{\theta_{\bm{z}_1}})$ PQC representation.
Step 5: $U(\bm{\theta_{\bm{z}_1}})$ is fed back as the input to the subsequent iteration, where the procedure goes to Step 2. In Appendix \ref{supp:pseudo}, we list the corresponding pseudo-code~\cite{supp}.

\textit{Quantum Algorithm ---}
If $\bm{z}$ were real-valued vectors, Eq.~\eqref{Q_loop} could be implemented using previous approaches~\cite{rebentrost2019quantum, cheng2024polynomial, li2020quantum}. However, the scenario where $\bm{z}$ resides in the complex domain, which remains unexplored. In this work, we extend previous results to address the complex case.

\begin{res} \label{res1}
Given $f(\bm{z})$ defined as a polynomial, mapping from $\mathbb{C}^{d+1} \rightarrow \mathbb{R}$, the effective gradient operator at $\bm{z}$ can be expressed as
\begin{equation}\label{eq:D}
\mathcal{D}(\bm{z}) = \operatorname{Tr}_{p-1}\left[\mathbb{I} \otimes \rho_{\bm{z}}^{\otimes p-1} \mathcal{M_D}\right],
\end{equation}
where $\rho_{\bm{z}} = \ket{\bm{z}}\bra{\bm{z}}$, and $\mathcal{M_D} = \sum_{k=1}^p \mathcal{P}_k \mathcal{F} \mathcal{P}_k$, with $\mathcal{P}_k$ denoting the permutation of the first and $k$-th subsystems in the $p$-fold tensor product $\ket{\bm{z}}^{\otimes p}$. 
\end{res}
\begin{proof}
    They are provided in Appendix \ref{supp:res1}~\cite{supp}.
\end{proof}
This result is formally equivalent to previous approaches. The distinction lies in our extension of the feasible region from the real domain to the complex domain. This is made possible by the pre-condition of a real-valued cost function and the adoption of a uniform learning rate for both the real and imaginary components. 

As a consequence, techniques based on linear combinations of unitaries (LCU) or Hamiltonian simulation, which implement Eq.~\eqref{Q_loop} in the real-valued domain, can now be extended to the complex domain while maintaining the efficiency of algorithms. 
For instance, in case where $\mathcal{F} = \sum_{\alpha=1}^k \otimes_{j=1}^p a_{\alpha, j} F_{\alpha,j}$, with $F_{\alpha,j}$ as unitary matrices and $a_{\alpha, j}$ as weights, and where $k$ is bounded, the algorithm—via the LCU framework—achieves a gate complexity of $\mathcal{O}(poly\log(d))$~\cite{cheng2024polynomial}. For more general cases, the Hamiltonian simulation approach can be used, only requiring an efficient method for extracting matrix elements. For example, the original Hamiltonian simulation protocol requires two oracles to query the coefficient matrix~\cite{low2017optimal}. More detailed implementation for realizing Eq.~\eqref{Q_loop} is provided in the Appendix \ref{supp:QGA}~\cite{supp}.

\textit{Classical Learning ---} 
To approximate $\ket{\bm{z}'}$ with $[U(\bm{\theta})]_t$, we need to solve Eq.~\eqref{C_loop}. $\boxed[.5]{C}$ works for it and is structured as two steps.

The first step utilizes a fixed PQC architecture, such that $[U(\bm{\theta})]_t = U(\bm{\theta_t})$. Analogous to standard variant quantum algorithms, $\bm{\theta}_t$ are optimized by minimizing the following expression:
\begin{equation} \label{eq:ind1}
    c_1 = 1 - \|\bra{\bm{z}'_t} U(\bm{\theta}_t)\ket{0}\|^2,
\end{equation}
where $\bm{\theta}_{\bm{z}_{t-1}}$ serves as the initial input parameter. The training strategy is regulated using a predefined threshold $\epsilon_0$. If the optimization of Eq.~\eqref{eq:ind1} yields a value below $\epsilon_0$, the procedure is halted, and the current parameters are accepted as $\bm{\theta}_{\bm{z}_t}$, thereby defining the PQC representation of $\ket{\bm{z}_t}$. If the threshold is not met, the process is directed to the second step.

In the second step, a shallow PQC, $T(\bm{\alpha})$, is appended to the existing $U(\bm{\theta}_t)$, followed by the minimization of Eq.~\eqref{eq:ind1}. $T(\bm{\alpha})$ is either designed based on the characteristics of $\mathcal{D}$ or generated through an architecture search algorithm. This modification enhances the expressibility of the circuit and increases the likelihood of minimizing 
\begin{equation}  \label{eq:ind2}        
   c_2 = 1-\|\bra{\bm{z}_{t}'} T(\bm{\alpha}) [U(\bm{\theta}_t)]\ket{0}\|^2 
\end{equation}
to exceed $\epsilon_0$. Upon it, we learn the added layer.
Once $c_2 \le  \epsilon_0$, we update $U(\bm{\theta_{\bm{z}_t}})$ with $T(\bm{\alpha})U(\bm{\theta}_{t-1})$ and go to the next iteration.
In Section \ref{simulation}, we demonstrate a RL-based protocol for this circuit learning, with Figure \ref{fig1}(b) showing the key steps.

In the NOM, the issue of gradient vanishing depicted in Section \ref{vanishing} is naturally circumvented. The expression $1 - \| \bra{\bm{z}_t'} U(\bm{\theta}_{\bm{z}_{t-1}}) \ket{0} \|^2$ can act as an indicator to distinguish between Eq.~\eqref{condition_1} and Eq.~\eqref{condition_2}, and is evaluated at the start of each iteration. Even when the gradient with respect to $\bm{\theta}$ vanishes—indicating that traditional optimization methods cannot further improve performance—this indicator allows us to detect whether the actual gradient is zero. Based on this evaluation, a $T(\bm{\alpha})$ term is appended and trained at a controllable computational cost as described earlier. Importantly, this indicator can be efficiently estimated, as it corresponds to the projective measurement of $U^\dagger(\bm{\theta})\ket{\bm{z}'_t}$ onto $\ket{0}$.
In summary, the protocol provides a clear criterion to determine whether a local minimum has been reached, guiding the evolution of the PQC even when $\|\nabla_{\bm{\theta}} f\| \approx 0$ but $\|\nabla_{\bm{z}} f\| \neq 0$.

\subsection{Numerical Results} 
\label{simulation}
To validate the effectiveness of NOM, we conducted simulations on two scenarios: (i) Max-cut problem on graphs with bounded degrees and (ii) the optimization of polynomial functions.

\begin{figure}[!ht]
   \centering  \includegraphics[width=1\columnwidth]{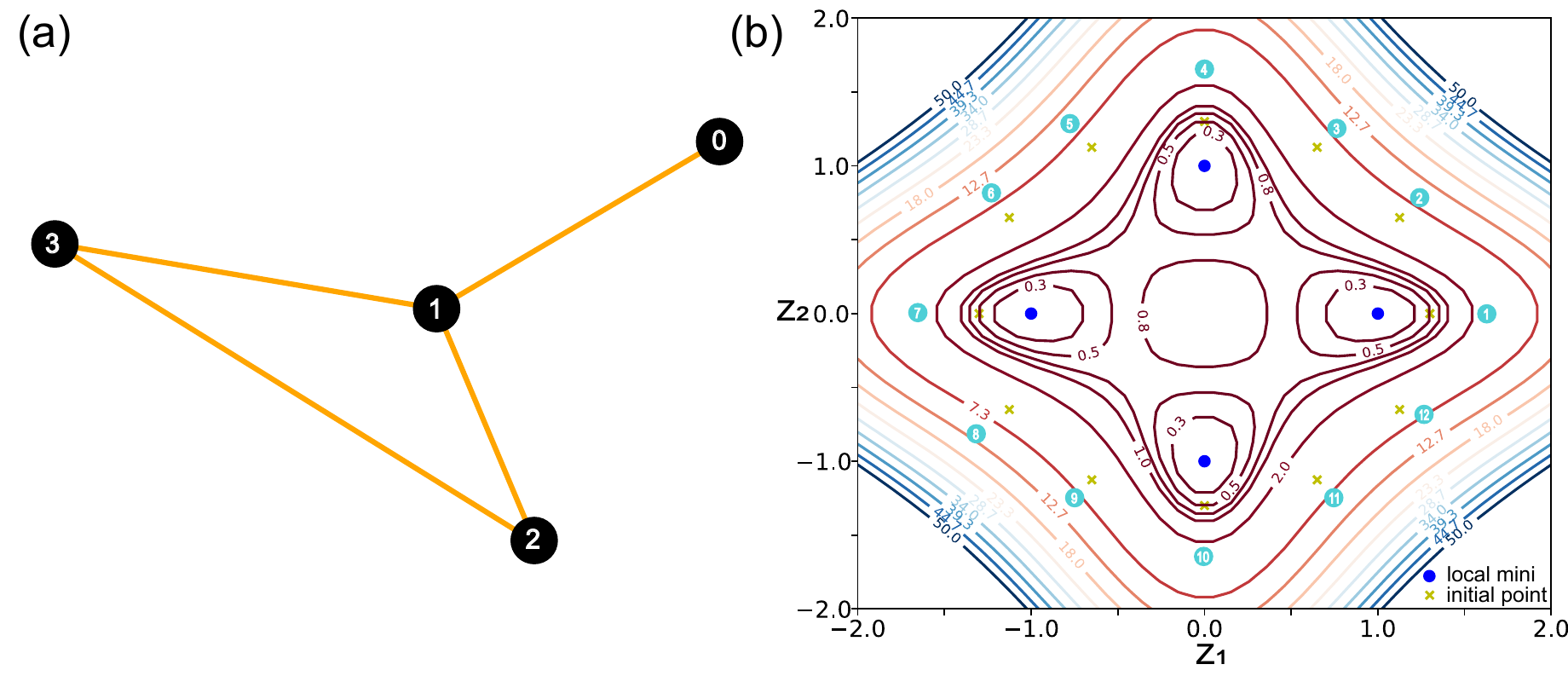}
   \caption{Illustration problems for numerical simulation. (a) is the graph for the MaxCut problem with bounded degrees, which includes 4 nodes and 4 edges. (b) is the contour plot of a polynomial within $z_3=0$ plane, where 4 local minimum are labeled as blue circles and 12 initial points are labeled as yellow crosses. } 
\label{num_problem} 
\end{figure}

The first scenario focus on the MaxCut problem with setting of a four-vertex graph with randomly generated edges, as shown in Figure~\ref{num_problem}(a). The objective function we aim to maximize is
\[
H_c = \sum_{\langle i,j \rangle} \frac{1 - \sigma_{z_i} \sigma_{z_j}}{2},
\]
where \(\sigma_{z_i}\) and \(\sigma_{z_j}\) represent the Pauli-Z operators acting on qubits \(i\) and \(j\), $\langle i,j \rangle$ exists only there is an edge. According to Edward Farhi~\cite{farhi2014quantum}, gradient-based algorithms can be used to learn a quantum state that approximates the solution to a given problem, with increased circuit depth leading to a better approximation ratio.
In our simulations, we first select 4 typical circuits from QAOA or VQA algorithms and fix the depth at 1, as specified in Appendix~\ref{supp:simu}. We then obtain the initial settings for our method by optimizing the parameters within these circuits until no further improvement can be obtained.
Figure~\ref{num_problem2}(a) presents the results, showing the cost function values as a function of the iteration number. Four cases, corresponding to four selected circuits, are labeled and tested.  Recalling that there are alternated $\boxed[.5]{Q}$ and $\boxed[.5]{C}$ parts in each iteration in our method, where $\boxed[.5]{Q}$ conduct the quantum gradient descent and $\boxed[.5]{C}$ conduct the PQC approximation of quantum state. The four solid `ideal' lines record the cost function value in each iteration when $\boxed[.5]{C}$ reconstruct states $\ket{\bm{z}'_t}$ ideally (direct record output without realizing $\boxed[.5]{C}$). All the data points in the four dashed `with PQC approximation' lines record the cost function value with optimal parameters in the $\boxed[.5]{C}$ part in each iteration(i.e.,  the gradients to the parameters approach zero). Notice that near-flat and jump are in four dashed lines. This occurs that even the indicator $c_1$ is not below the criteria, we accumulate difference to establish a better layered circuit. Every jump means we add a layer and process RL-based method. Notice that the near-flat regions and jumps in the four dashed lines, which corresponds to the addition of a new layer or not, which is followed by the application of the RL-based method. 
Notice the  jumps and near-flat regions in the four dashed lines, which correspond to whether a new layer is added or not, followed by the application of the RL-based method.
We can see that whenever the `indicator' is non-zero, the cost function value will not reach its real maximum in each iteration, although the $\boxed[.5]{C}$ part successfully reaches the optimal parameters. This highlights the limitation of conventional methods with a fixed ansatz, which may fail to reach the optimal solution if the circuit is not properly configured. With the help of the $\boxed[.5]{Q}$ part, each case successfully converges to the target value, demonstrating our method’s capability in mitigating gradient vanishing issues and reaching the optimal cost function.

The second scenario involves the optimization of polynomial. We denote this cost function as 
\[
f(Z) = Z^{\otimes 2} A Z^{\dagger \otimes 2},
\]
where \(Z = (1, z_1, z_2, z_3)\) and \(A=\sigma_x^{\otimes 4}+\sigma_y^{\otimes 4}+\sigma_z^{\otimes 4}\) is a coefficient matrix derived from the polynomial expression~\cite{supp}. 
It is observed that the cost function features isolated local minima within the region \((z_1,z_2,z_3) \in [-2, 2]^{\otimes 3}\), and our simulation focuses on optimizing within this region. In Figure~\ref{num_problem}(b), we present a case where \(z_3 = 0\). We start with 12 initial points, each positioned on a circle with a radius of 1.3. A parameterized circuit is then used to approximate these points via a state-to-state method, serving as the initial settings for our method.
The simulation results for the cost function values are shown in Figure~\ref{num_problem2}(b). Twelve cases, corresponding to twelve initial points, are labeled and tested. For all the initial values, the cost function successfully converges to the target value. In this scenario, feasibility is ensured by using a larger set of initial seeds, each of which is driven to its respective local minimum. The reason for the deviations remains the same as in the first scenario.

\begin{figure}[!ht]
   \centering  \includegraphics[width=1\columnwidth]{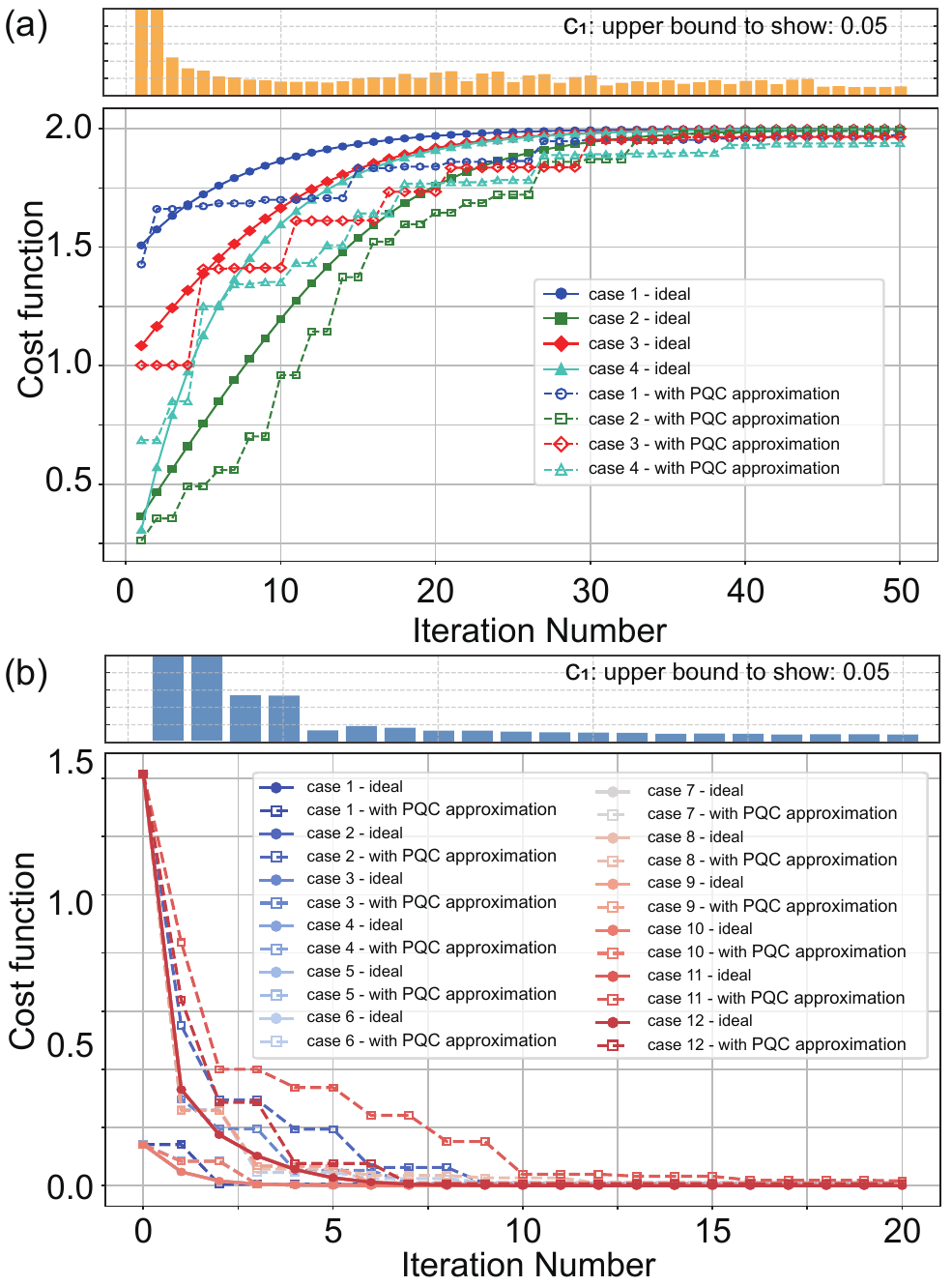}
   \caption{Numerical results of the quantum-classical learning protocol for (a) Maxcut problem and (b) polynomial optimization. Solid markers and lines denotes the ideal case. Dashed markers and lines denotes the cases where the classical loop is approximated by PQC generated by our RL method with certain accuracy limit (fidelty=0.998).} 
\label{num_problem2} 
\end{figure}

We outline our numerical procedure here. 
The simulation is carried out according to the described protocol. We set learning rate to be $\xi = 0.2$ for both examples. For each iteration, a quantum gradient algorithm is conducted with a cost function defined by either a Hamiltonian or polynomial coefficients. For the classical part of learning algorithm, we introduce a RL algorithm for learning $T(\bm{\alpha})$. The RL algorithm implement a PPO algorithm that learns a policy to approximate the output state from quantum gradient step. After the policy is learned from multiple simulated episodes, the quantum circuit is constructed step by step using a pre-selected universal gate set. The output state is represented by the learned PQC that acts on the input state.  The details of the algorithm is described in the Appendix \ref{supp:RL}~\cite{supp}.

Remarkably, the PQC representation is an approximation of the output state of the quantum gradient algorithm. To investigate how  the algorithm is affected by the inaccuracy of states,we additionally run the algorithm with artificially added noise. 
Shown in Appendix \ref{supp:simu}~\cite{supp}, we examine the influence of disturbance magnitude on each iteration.
The results indicate that the protocol remains robust under a certain level of disturbance, but if the approximation quality deteriorates significantly, convergence is disrupted. This robustness can be attributed to the inherent stability of the gradient, as discussed in Appendix~\ref{supp:error}~\cite{supp}. However, the influence of disturbance on the final convergence, and how to effectively measure it, remains an open question that could be explored in future work.

\section{Feasibility}

\textit{Managing the learning cost for added layer --- }
If the first step of classical learning proves not working, our method will shift to minimize $1-\|\bra{\bm{z}_{t}'} T(\bm{\alpha}) \ket{\bm{z}_{t-1}}\|^2$, where $\ket{\bm{z}'_t} = \ket{\bm{z}_{t-1}} - \xi \mathcal{D} \ket{\bm{z}_{t-1}}$. As $\|\ket{\bm{z}_t}-\ket{\bm{z}'_{t}} \| \approx \xi \|\mathcal{D}\|$, this searching region is bounded by $\xi$. Given that $\ket{\bm{z}'_{t}}$ is close to $\ket{\bm{z}_{t-1}}$, the initial configuration of $T(\bm{\alpha})$, consisting of alternating layers of single- and two-qubit gates, can be set near the identity. 

Notably, when $\xi \|\mathcal{D}(\bm{z})\|$ is small, the circuit can remain shallow, with a depth of $\mathcal{O}(\xi^2 \|\mathcal{D}(\bm{z})\|^2)$~\cite{campbell2019random}. Consequently, we can limit the depth of $T(\bm{\alpha})$ and continue the architecture search using manageable classical resources.

\textit{Barren plateaus in classical learning part ---}
Next we analyze barren plateaus in learning $T(\bm{\alpha})$.
We denote that $T(\bm{\alpha}) = \prod_{j=1}^{L} T_j(\alpha_j)$, where $L$ is finite. The partial derivative to the $k$-th parameter in Eq.~\eqref{eq:ind2} is 
\begin{equation}\label{bp_gradient1}
    \frac{\partial c_2}{\partial \alpha_k } = i \bra{\bm{z}}T^{\dagger}_{-} \left[ V_k, ~ T^{\dagger}_{+} \ket{\bm{z}'}\bra{\bm{z}'} T_{+} \right] T_{-} \ket{\bm{z}},
\end{equation}
where $V_k$ denotes the derivative of a single-layer circuit $T_k(\alpha_k)$, $[\cdot]$ represents the commutator, and $T_{-} = \prod_{j=1}^{k} T_j(\alpha_j), T_{+} = \prod_{j=k+1}^{L} T_j(\alpha_j)$.
Initially, $T(\bm{\alpha}) = T_{+}T_{-} = I$, which simplifies Eq.~\eqref{bp_gradient1}  to
\begin{equation} \label{bp_gradient2}
    i\bra{\bm{z}} \left[\ket{\bm{z}'}\bra{\bm{z}'}, ~ V_k \right]\ket{\bm{z}},
\end{equation}
with $T_{+} = T_{-}^{\dagger}$ initialized as the identity. Eq.~\eqref{bp_gradient2} equals zero only when the target state commutes with the observable, which is generally not the case~\cite{grant2019initialization}. When $T_{+}$ and $T_{-}$ are not initialized from identity, the finite search region centered around the identity ensures that the major part, i.e., Eq.~\eqref{bp_gradient2} is preserved, thereby mitigating barren plateaus.

During the optimization process, when $\xi$ is constrained, the depth of $T(\bm{\alpha})$ remains bounded, and the search region remains near the identity operation. Consequently, the partial derivative predominantly behaves as described in Eq.~\eqref{bp_gradient2}. Further details are provided in the Appendix~\ref{supp:gradient_vanishing}~\cite{supp}.

\textit{Computational cost of the NOM ---}
Notations for relevant variables are consistent with Section~\ref{theory}. 

First, we consider the sample complexity of each iteration.
Based on prior work~\cite{rebentrost2019quantum, li2020quantum, cheng2024polynomial}, the sample complexity for the quantum algorithm in each iteration is given by $\mathcal{O}\left( kp \log(kp)/ \epsilon_L^2+ (1+\xi)^{2}/\delta_L^2 \right))$, where $\epsilon_L$ and $\delta_L$ are bounded errors for statistical estimations using the LCU method. Alternatively, it can be expressed as $\mathcal{O}\left(p^3||\mathcal{F}||^2_{max}(1+\xi)^{2}/\epsilon_H \right)$, where $\epsilon_H$ is the bounded error for the Hamiltonian simulation method. 
For the classical learning part, it requires repeatedly querying the state output by quantum gradient algorithm. 
As only the probability onto $\ket{0}$ is required to known, $\mathcal{O}(1/\epsilon_c^2)$ copies of states is required, where $\epsilon_c$ is produced for probability property of measurement.

Therefore, the total sample complexity for one iteration of the NOM is,
\begin{eqnarray}\label{sample_com}
    \mathcal{O} \left(\frac{1}{\epsilon_c^2}\cdot poly(p,\xi, C_{\mathcal{F}}, \frac{1}{\epsilon} )  \right),
\end{eqnarray}
where $C_{\mathcal{F}}$, $\epsilon$ are rough parameters, which relate to $\mathcal{F}$ and induced error.

Second, the time required for each iteration is divided into two parts: the time to run the quantum algorithm and the time for classical training. From~\cite{rebentrost2019quantum, li2020quantum, cheng2024polynomial}, the time complexity to process the quantum gradient is,
\begin{eqnarray}\label{sample_cost}
    \mathcal{O}\left(poly(p,\log(d),\xi,C_{\mathcal{F}}, \frac{1}{\epsilon})\right).
\end{eqnarray}
For the classical learning part, it is split into quantum and classical parts. The cost function is evaluated using a quantum device, whose time consumption is linear with the depth of circuit multiplied by sample cost in Eq.~\eqref{sample_cost}, while other operations are executed on a classical processor. Since the depth of the learning circuit is bounded by $\xi \|\mathcal{D}(\bm{z})\|$ (shown in Method) and a limited quantum gate set, we hypothesize that the complexity of designing quantum circuits can be constrained. Therefore, the time to query the quantum circuit is linear with respect to the circuit depth and the number of learning episodes.

Finally, each input sample is encoded with a PQC. Increasing iterations in the NOM would increase the cost by multiplying it with number of iterations.

\section{Discussion}
How can we generate a PQC for a given target? Unlike previous approaches that construct and optimize a classically accessible parameter space, the approach presented in this paper explores the Hilbert space using quantum gradients. This is achieved through the NOM that integrates both quantum and classical algorithms.
Our analysis shows that this protocol effectively mitigates certain instances of gradient vanishing and performs well in achieving target objectives. Meanwhile, the classical learning algorithm within the protocol is free from barren plateaus, a direct result of the small search region and the finite depth of the training circuit. This also leads to another significant advantage—affordable computational costs. In addition, the protocol presented in this paper offers contributions to other fields as well.

\textit{Impact on Quantum Gradient Algorithms --- }
From an algorithmic perspective, the proposed protocol enhances the quantum gradient algorithm for polynomial optimization. The challenge of optimizing a polynomial traces back to Hilbert's seventeenth problem, which asks whether a multivariate polynomial can be expressed as a sum of squares of rational functions. Current quantum algorithms typically rely on cascade circuits, each with multiple inputs and a single output. These algorithms face two significant challenges: first, the number of iterations required is often unknown, making it difficult to predict the necessary sample size; second, the exponential growth in sample consumption with growing iterations~\cite{2019-Rebentrost-Qgradient, cheng2024polynomial}. 
The NOM, which utilizes PQCs to approximate temporal quantum states, addresses these challenges while preserving the benefits of earlier frameworks, such as a time complexity of $\mathcal{O}(poly(\log (d)))$ per iteration. Further details, particularly from the perspective of the quantum gradient algorithm in optimization, are provided in the Appendix~\cite{supp}.

\textit{Extending the protocol ---} 
We consider the case of analytic $f(\bm{z})$ on domain $\sigma$, denoted as $f\in \mathbb{H}(\sigma): \mathbb{C}^{d+1} \rightarrow \mathbb{C}$. 
According to Cauchy–Riemann conditions, for any $\bm{z}\in \sigma$, the total derivative is given as $df =  \frac{\partial f}{\partial z} dz$. Formally, we can express the gradient descent iteration as, 
\begin{eqnarray}\label{eq:iter2}
    z' = z - \frac{\partial f}{\partial z} \Delta z.
\end{eqnarray}
We observe that $\mathcal{D}(\bm{z}) \ket{\bm{z}^{*}}$ produces the same effect as ${\partial f}/{\partial z}$ and the emphasis is to prepare $\ket{z^{*}}$. Although realizing arbitrary $\ket{z^{*}}$ is challenging, the state has its PQC approximation. By simply modifying the rotations along the x- and z-axes to their inversions, we can obtain the complex form. Upon it, this iterative process can be implemented in a analytic complex domain.

In summary, the learning of PQCs has long faced significant challenges due to trainability.  Our proposed protocol leverage quantum gradient, advancing both quantum optimization techniques and the learning of parameterized quantum circuits. Given the broad applications of PQCs in both near-term and future fault-tolerant quantum devices, this method is expected to offer unique advantages in terms of trainability and performance, offering significant benefits for quantum circuit synthesis and quantum machine learning problems.

\bibliographystyle{unsrt}
\bibliography{qG}


\begin{widetext}

\newpage
\counterwithin{figure}{section}
\appendix

\section{Details on Gradient Vanishing}
\label{supp:gradient_vanishing}
\textit{Proof for Result 1 ---}  
To get Result 1, we first introduce Lemma \ref{prob_orthon}, which give a probability bound that two random vector have overlap. 
\begin{lemma}\label{prob_orthon}
    Suppose $\bm{u}$, $\bm{v}$ be two $d$-dimension random vector whose elements are chosen from $\mathcal{N}(0,1/d)$, for some $\epsilon \in (0,1)$ there is 
    \begin{eqnarray}
        \text{Pr}(||\langle \bm{u}|\bm{v}\rangle||\geq \epsilon ) \leq 4 e^{-\frac{d\epsilon^2}{8}}
    \end{eqnarray}
    where $\langle \cdot | \cdot \rangle$ denotes the inner product and $||\cdot ||$ denotes the absolute value.
\end{lemma}
\begin{proof}
     Since $\bm{u}$ and $\bm{v}$ are both chosen randomly from $\mathcal{N}(0,1/d)$, so $\bm{w}=[w_1,\cdots, w_d]=\frac{\bm{u}+\bm{v}}{\sqrt{2}}\sim \mathcal{N}(0,1/d)$,too. 
     
     We have the probability 
     \begin{eqnarray}\label{prob_proof}
         &&\text{Pr}(||\bm{w}||^2-1\geq \epsilon)=\text{Pr}(e^{\lambda(||\bm{w}||^2-1)}\geq e^{\lambda\epsilon}) \nonumber \\
         &&\leq \mathop{min}_{\lambda >0}   e^{-\lambda \epsilon} \mathbb{E}[e^{\lambda (||\bm{w}||^2-1)}] \nonumber \\
         &&=\mathop{min}_{\lambda >0}   e^{-\lambda (\epsilon+1)} \mathbb{E}[e^{\lambda (||\bm{w}||^2)}] \nonumber \\
         &&=\mathop{min}_{\lambda >0}   e^{-\lambda (\epsilon+1)} \prod_i \mathbb{E}[e^{\lambda (w_i^2)}] \nonumber \\
         &&=\mathop{min}_{\lambda >0}   e^{-\lambda (\epsilon+1)} (\frac{d}{d-2\lambda})^{d/2} \nonumber \\
         &&\leq e^{d(\log(1+\epsilon)-\epsilon)/2} \leq e^{-d \epsilon^2/8}
     \end{eqnarray}
     for any positive $\lambda$ and $\mathbb{E}[\cdot]$ be the expectation. The second line in above equation uses the fact that 
     \begin{eqnarray}
         \mathbb{E}[x]=\int_0^{\infty} x p(x) \geq \int_a^{\infty} a p(x)=aP(x\geq a).
     \end{eqnarray}
     Similarly to  Eq.~\eqref{prob_proof}, we have $\text{Pr}(1-||\bm{w}||^2\geq \epsilon)\leq e^{-d \epsilon^2/8} $, too. And therefore,  \begin{eqnarray}
         \text{Pr}(\langle u|v\rangle \geq \epsilon)&&\leq \text{Pr}(||\bm{w}||^2-1 \geq \epsilon)+\text{Pr}(1-||\bm{w}||^2\geq \epsilon) \nonumber \\
         &&\leq 2 e^{-d \epsilon^2/8},
     \end{eqnarray} 
      \begin{eqnarray}
         \text{Pr}(-\langle u|v\rangle \leq 2 e^{-d \epsilon^2/8},
     \end{eqnarray} 
     too. So $\text{Pr}(||\langle \bm{u}|\bm{v}\rangle || \geq \epsilon) \leq 4 e^{-d \epsilon^2/8} $ is proved.
\end{proof}

Then we start the proof of Result 2.
\begin{proof}

In our case, we focus on 
\begin{equation}\label{derivatives2}
        \nabla_{\bm{\theta}} f = \frac{\partial \bm{z}}{\partial \bm{\theta}} \cdot \nabla_{\bm{z}} f, \quad \text{where}\quad \frac{\partial{\bm{z}}}{\partial \bm{\theta}} = \begin{pmatrix}
             \frac{\partial z_1}{\partial \theta_1} & \dots & \frac{\partial z_d}{\partial \theta_1} \\
             \vdots & \ddots & \vdots \\
             \frac{\partial z_1}{\partial \theta_m} & \dots & \frac{\partial z_d}{\partial \theta_m} 
    \end{pmatrix},
\end{equation}
which is introduced in manuscript. We want to study the how large the norm of $\nabla_{\bm{\theta}} f$ can be.

Statistically, the gradient $\bm{\nabla_z}f$ can be viewed as a vector randomly chosen from $d$ dimension complex space. We can also denote the $i$-th row of the matrix $\partial \bm{z}/\partial \bm{\theta}$ as a $d$ dimension vector $v_i$, too. 

First, let us observe $\langle v_i | \bm{\nabla_z}f \rangle$. 
Here we represent the complex variables $\bm{v}_i$ and $\bm{\nabla_z}f$ in real and imaginary parts as $\bm{v}_i=\partial \bm{z}/\partial \theta_i =\alpha_i+i \beta_i$, and $\bm{\nabla_z}f =\gamma + i\delta$, where $\alpha,\beta,\gamma,\delta$ can be treated as independent $d$-dimension  random real vectors. 

Given assumptions that the norm of $\bm{\nabla_z}f$ and $v_i$ are bounded.
We have that $||\alpha||,||\beta||,||\gamma||,||\delta||$ is nearly independent to the Hilbert space dimension $d$. If their elements are  sampled from independent identically distribution $\mathcal{N}(0,K)$,  the independence of $||\alpha||=\sqrt{\sum_{i=1}^d \alpha_i^2}$ to $d$ indicates that $\alpha_i \propto 1/\sqrt{d}$ and therefore $K\sim C/d$ for some constant $C$. That is $\alpha,\beta,\gamma,\delta$'s elements are chosen randomly from $\mathcal{N}(0, C/d)$ with $C$ being some bounded positive value so that the norm is bounded. Therefore $\alpha/\sqrt{C},\beta/\sqrt{C},\gamma/\sqrt{C},\delta/\sqrt{C}$ are random vectors whose elements are sampling from $\mathcal{N}(0,1/d)$. Then we have 
\begin{eqnarray}
    \langle \bm{v}_i|\bm{\nabla}f\rangle = \langle \alpha|\gamma\rangle -\langle \beta|\delta\rangle +i (\langle \alpha|\delta\rangle-\langle \beta|\gamma\rangle
\end{eqnarray}
and the probability 
\begin{eqnarray}\label{prob_comp_ortho}
    &&\text{Pr}(||\langle \bm{v}_i|\bm{\nabla}f\rangle||\geq \epsilon C) \nonumber \\
    &&\leq 1- \text{Pr}(||\langle \alpha|\gamma\rangle -\langle \beta|\delta \rangle||\leq  \frac{\epsilon}{\sqrt{2}} C ) \nonumber \\
    &&\times \text{Pr}(||\langle \alpha|\delta \rangle-\langle \beta|\gamma \rangle|| \leq \frac{\epsilon}{\sqrt{2}} C)\nonumber \\
    &&= 1- \text{Pr}(|\langle (\alpha, \beta)|(\gamma,-\delta)\rangle| \leq  \frac{\epsilon }{\sqrt{2}} C) \nonumber \\
    && \times  \text{Pr}(||\langle (\alpha,\beta)|(\delta,-\gamma)\rangle|| \leq \frac{\epsilon}{\sqrt{2}} C )  \nonumber \\
     &&\sim 1- \text{Pr}(||\langle (\frac{\alpha}{2\sqrt{C}}, \frac{\beta}{2\sqrt{C}})|(\frac{\gamma}{2\sqrt{C}},-\frac{\delta}{2\sqrt{C}})\rangle|| \leq  \frac{\epsilon}{4\sqrt{2}} ) \nonumber \\
    && \times  \text{Pr}(||\langle (\frac{\alpha}{2\sqrt{C}},\frac{\beta}{2\sqrt{C}})|(\frac{\delta}{2\sqrt{C}},-\frac{\gamma}{2\sqrt{C}})\rangle|| \leq \frac{\epsilon}{4\sqrt{2}}  )  \nonumber \\
    &&\leq 1-(1-4e^{-\frac{d \epsilon^2}{128}})(1-4e^{-\frac{d \epsilon^2}{128}}) \nonumber \\
    &&\leq 8 e^{-\frac{d\epsilon^2}{128}}
\end{eqnarray}
where we have the notation $(\alpha,\beta)$ means a $2d$-dimension vector whose first $d$-dimension half part is $\alpha$ and the last half is $\beta$, similarly the notation $(\gamma,-\delta)$ and $(\delta,-\gamma)$. In the last two lines in Eq.~\eqref{prob_comp_ortho} we use the fact that $||\alpha,\beta||=||\bm{v}_i||$ and $||\gamma,-\delta||=||\delta,-\gamma||=||\bm{\nabla_z}f||$, as well as the result of Lemma.\ref{prob_orthon}.
The result of Eq.~\eqref{prob_comp_ortho} means in high dimension $d$, the vector $\bm{v}_i$ and $\bm{\nabla_z}f$ will orthogonal to each other with probability near to 1.

we estimate the probability of the norm $||\bm{\nabla_\theta}f||=\sqrt{\sum_{i=1}^{poly(\log(d))}||\langle \frac{\partial{\bm{z}}}{\partial \bm{\theta}_i} | \bm{\nabla_z}f\rangle ||^2}$ be larger than $\epsilon poly(\log(d)) C$, as 
\begin{eqnarray}
    &&\text{Pr}(||\bm{\nabla_\theta}f||\geq  \epsilon  C \sqrt{poly(\log(d))}) \nonumber \\
    && \leq 1- \prod_{i=1}^{poly(\log(d))}(1-\text{Pr}(||\langle \bm{v}_i|\bm{\nabla_z}f\rangle ||\geq \epsilon  C)) \nonumber \\
    && = 1-(1-8e^{-\frac{d\epsilon^2}{128}})^{poly(\log(d))}\nonumber \\
    &&\leq 8poly(\log(d)) \cdot  e^{-\frac{d \epsilon^2}{128}}
\end{eqnarray}

Remarkably, we make a variable substitution to transfer $C$ and $\log(d)$ to the right side of formula.
\end{proof}

\section{the Pseudo-code of NOM}
\label{supp:pseudo}

In this section, the pseudo-code of NOM is presented in Table~\ref{algo_opt}, corresponding to the theory and Figure 1 of the manuscript.

\floatname{algorithm}{Table}
\begin{figure}[H] 
\begin{algorithm}[H]
   \caption{The NOM.}
   \label{algo_opt} 
   \hspace{\algorithmicindent} 
   \textbf{Input:} an initial ansatz $U(\bm{\theta}_{z_0})$ with its structure and therein parameters to be optimized, $m$, the maximum iteration number, $\epsilon$, the terminate threshold, and $\xi$, the learning rate. 
   \begin{spacing}{1.0}
      \begin{algorithmic}[1]
      \State Set $\Delta f\leftarrow 1$
      \While{$t-1\in [m]\cap \Delta f\ge \epsilon$ }
      \Function{quantum\_algorithm\_part}{$U(\bm{\theta}_{z_{t-1}})$}
      \State \Return $\ket{\bm{z}_{t}'}\leftarrow U(\bm{\theta}_{z_{t-1}})\ket{0}-\xi \mathcal{D}U(\bm{\theta}_{z_{t-1}})\ket{0}$ \Comment{$\ket{\bm{z}_{t-1}}=U(\bm{\theta}_{z_{t-1}})\ket{0}$}
      \EndFunction
      \State Calculate $\Delta f \leftarrow 1 - \|\bra{\bm{z}'_t} U(\bm{\theta}_{z_{t-1}})\ket{0}\|^2$
      \Function{classic\_learning\_part}{$\ket{\bm{z}_{t}'}$, $U(\bm{\theta}_{z_{t-1}})$}
      \State \Return $U(\bm{\theta}_{z_{t}})$
      \EndFunction 
      \EndWhile
      \State \textbf{return:} $U(\bm{\theta}_{z_{t}})$
      \Algphase{Steps of sub-function of quantum loop (depicted in Hamiltonian Simulation method)}
      \Require The principal register is initialized with $U(\bm{\theta}_{z_{t-1}})$, and the three other ancillary systems, $\ket{0}_{lcu}\ket{0}_{d}\ket{0}_e$
      \State Driving the entire system through the circuit of quantum gradient algorithm in Figure~\ref{circuit_quantum_loop}.
      \State Post-select on $\ket{0}_{lcu}\ket{0}_{d}$ and the principal register is output, i.e., $\ket{\bm{z}'_{t}}$
      \Algphase{Steps of sub-function of classic loop}
      \Require $U(\bm{\theta}{z_{t-1}})$, and $\ket{\bm{z}'_{t}}$ is repeatedly produced
      \Function{1st~step}{$U(\bm{\theta}{z_{t-1}})$, $\ket{\bm{z}'_{t}}$}
      \State Transport $\ket{\bm{z}'_{t}}$ to the inverse circuit of $U(\bm{\theta}_{z_{t-1}})$
      \State Measure the possibility on $\ket{0}\bra{0}$
      \State \textbf{return:} $\bm{\theta}$, which is tuned to max the aforementioned possibility
      \EndFunction
      \State If the cost function satisfies the specified criteria, we terminate the process; otherwise, we proceed to the next function
      \Function{2nd~step}{$U(\bm{\theta}{z_{t-1}})$, $\ket{\bm{x}'}$} \Comment{This step can be realize with various method, we refer to reinforcement learning and depict it later}
      \State Transport $\ket{\bm{z}'_{t}}$ to a initialized a slice of circuit $T(\bm{\alpha})$ and the inverse circuit of $U(\bm{\theta}{z_{t-1}})$
      \State Measure the possibility on $\ket{0}\bra{0}$
      \State \textbf{return:} $T(\bm{\alpha})$, which is tuned to max the aforementioned possibility
      \EndFunction
      \State $U(\Tilde{\bm{\theta}})\leftarrow U(\bm{\theta}) ~\mbox{or} ~ T(\bm{\alpha}) U(\bm{\theta})$ is returned
      \end{algorithmic}
     \end{spacing}
\end{algorithm}
\end{figure}


\section{Quantum gradient algorithm in complex domain}
\label{supp:res1}
First, we review Result 2.
Given that $f(\bm{z})$ is defined as a polynomial, mapping from $\mathbb{C}^{d+1} \rightarrow \mathbb{R}$, the effective gradient operator at $\bm{z}$ can be expressed as
\begin{equation}
\mathcal{D}(\bm{z}) = \operatorname{Tr}_{p-1}\left[\mathbb{I} \otimes \rho_{\bm{z}}^{\otimes p-1} \mathcal{M_D}\right],
\end{equation}
where $\rho_{\bm{z}} = \ket{\bm{z}}\bra{\bm{z}}$, and $\mathcal{M_D} = \sum_{k=1}^p \mathcal{P}_k \mathcal{F} \mathcal{P}_k$, with $\mathcal{P}_k$ denoting the permutation of the first and $k$-th subsystems in the $p$-fold tensor product $\ket{\bm{z}}^{\otimes p}$. 

\begin{proof}
As \( f(z, \bar{z}) \) is real valued at every moment, it can be regarded as a real function $f(x,y)$ with two real variables, 
\begin{eqnarray}
    z = x + iy, \quad \bar{z} = x - iy,
\end{eqnarray} 
where \( x, y \in \mathbb{R} \).
By definition of the gradient descent algorithm in the real domain, the iterative formula can be  
\begin{eqnarray}
    \left(x', y'\right) = \left(x - \frac{\partial f}{\partial x} \Delta x,  y - \frac{\partial f}{\partial y} \Delta y \right)
\end{eqnarray}
where $x', y'$ are updated variables. This implies that an updated $z'$ can be expressed as,   
\begin{eqnarray}\label{eq:proof2}
    z' = x + iy - \left[ \frac{\partial f}{\partial x} \Delta x + i \frac{\partial f}{\partial y} \Delta y \right].
\end{eqnarray}
Meanwhile,
\begin{eqnarray}
    \frac{\partial f}{\partial \bar{z}} = \frac{1}{2} \left( \frac{\partial f}{\partial x} + i \frac{\partial f}{\partial y} \right). 
\end{eqnarray}
When the learning rate is set to \( \Delta x = \Delta y = \xi/2 \), we obtain, 
\begin{eqnarray}
    z' = z - \xi \frac{\partial f}{\partial \bar{z}},
\end{eqnarray}
which is the gradient descent formula used in our case.

On the other side, we substitute $ f(\bm{z}) =\bm{z}^{\dagger \otimes p} \mathcal{F} \bm{z}^{\otimes p} =f(\bm{z}, \bm{\bar{z}}) = \bm{\bar{z}}^{T \otimes p} \mathcal{F} \bm{z}^{\otimes p}$ into partial derivative, then
\begin{eqnarray}
    \frac{\partial f}{\partial \bar{z}}=\sum_{\alpha=1}^k \sum_{i=1}^p c_{\alpha,i}(\bm{z}) F_{\alpha,i} \bm{z} \doteq \mathcal{D}(\bm{z})\bm{z},
\end{eqnarray}
This holds on the fact that 
\begin{eqnarray} \label{eq:decomp}
   \mathcal{F}=\sum_{\alpha=1}^k  \otimes_{j=1}^p a_{\alpha, j} F_{\alpha,j} 
\end{eqnarray}
such that $c_{\alpha,i}(\bm{z})=[\prod_{j=1}^{p} (a_{\alpha, j} \bm{z}^{\dagger}F_{\alpha,j}\bm{z})]/\bm{z}^{\dagger}F_{\alpha,i}\bm{z}$. Therefore, 
\begin{eqnarray}
    \mathcal{D}(\bm{z})=(\mathbb{I} \otimes \bm{z}^{\dagger \otimes p-1})\sum_{k=1}^p P_k A P_k (\mathbb{I} \otimes\bm{z}^{\otimes p-1}),
\end{eqnarray}
where $\mathcal{P}_k$ permutes the first and the $k$-th subsystems. Up to this point, the problem formulates as a problem encountered in~\cite{rebentrost2019quantum, li2020quantum, cheng2024polynomial}.

After encoding $\bm{z}$ into a quantum state, we can get 
\begin{equation}
 \mathcal{D}(\bm{z})=Tr_{p-1}\big[\mathbb{I} \otimes \rho_{\bm{z}}^{\otimes p-1}\mathcal{M_D}\big], 
\tag{\ref{eq:D}}   
\end{equation}
which can be realized by LCU or Hamiltonian simulation-based methods with the time complexity of $\mathcal{O}(poly(p,\mathcal{F}, \log d))$. 
\end{proof}

\section{Implementation of Quantum Gradient}
\label{supp:QGA}

This section provides detailed instructions on the implementation of the quantum gradient algorithm.
That is, to implement the following iterative steps:
\begin{eqnarray}\label{iteration_eq}
   \ket{\bm{z}} \leftarrow \ket{\bm{z}} \pm \xi \mathcal{D} \ket{\bm{z}},
\end{eqnarray}
where the parameter-dependent operator \(\mathcal{D}(\bm{z})\) is the gradient operator, and \(\ket{\bm{z}}\) is assumed to encode the variable. 

\subsection{Quantum Gradient Part with Oracles to query Coefficient Matrix }

\textit{Algorithmic Method for Polynomial-type Optimization ---}
For polynomial-type optimization, we introduce the quantum gradient algorithm, utilizing oracles for efficient querying of the coefficient matrix. This is achieved through the quantum matrix inversion method via Hamiltonian simulation~\cite{rebentrost2019quantum, li2021optimizing, gao2021quantum}.
Specifically, \(\mathcal{D}\) is defined as,
\begin{eqnarray}\label{DandH}
\mathcal{D} = \operatorname{Tr}_{p-1}\left[\mathbb{I} \otimes \rho_{\bm{z}}^{\otimes p-1} \mathcal{M_D}\right].
\end{eqnarray}
The relevant notations are as follows,
\begin{eqnarray}\label{M_expansion}
  \rho_{\bm{z}} = \ket{\bm{z}}\bra{\bm{z}}, \nonumber \\
  \mathcal{M_D} = \sum_{k=1}^p \mathcal{P}_k \mathcal{F} \mathcal{P}_k,
\end{eqnarray}
where \(\mathcal{M}_{\mathcal{D}}\) and \(\mathcal{F}\) have relations by \(\mathcal{P}_k\), which swaps the first and the \(k\)-th subsystems of the entire space \(\ket{\bm{z}}^{\otimes p}\), as
\[
\mathcal{S}\ket{\bm{z}^1}\otimes\dots\ket{\bm{z}^{k-1}} \otimes \ket{\bm{z}^k}\otimes\ket{\bm{z}^{k+1}}\dots =\ket{\bm{z}^k}\otimes \cdots \ket{\bm{z}^{k-1}}\otimes \ket{\bm{z}^{1}}\otimes\ket{\bm{z}^{k+1}}\cdots.
\]

Eq.~\eqref{iteration_eq} is proven to be implementable on a quantum computer, providing an exponential speed-up in terms of operation complexity. These complexities are listed as:
\begin{eqnarray}
    \mathcal{O}(\log(d))~\text{(space complexity)}, \quad \mathcal{O}(poly(\log(d)))~\text{(time complexity)}.
\end{eqnarray}
The procedure for conducting the quantum gradient algorithm is outlined in Table \ref{algorithm2}.

\floatname{algorithm}{Algorithm}
\begin{algorithm}[H]
   \caption{The quantum gradient algorithm}
   \label{algorithm2} 
   \hspace{\algorithmicindent} 
   \begin{spacing}{1.1}
      \begin{algorithmic}[1]
      \State Initialization:
      \begin{eqnarray}\label{eq:qgalcu}         
      \ket{0}_{lcu}\ket{0}_{d}\ket{0}_e\ket{0}\rightarrow \ket{0}_{lcu}\ket{0}_{d}\ket{0}_e\ket{\bm{z}}-\xi \ket{1}_{lcu}\ket{0}_{d}\ket{0}_e\ket{\bm{z}}
      \end{eqnarray}
      where $\ket{0}_{lcu}$ is applied by
\begin{eqnarray}
    R_y(\xi)= \begin{pmatrix}
        1/\sqrt{1+\xi^2} & \xi/\sqrt{1+\xi^2}\\
        -\xi/\sqrt{1+\xi^2} & 1/\sqrt{1+\xi^2}
    \end{pmatrix}.
\end{eqnarray}

      \State Matrix Inversion on $\ket{1}_{lcu}$: 
      \begin{itemize}
          \item First, Phase estimation of $e^{-i\mathcal{D} t}$,
      \begin{eqnarray}        \ket{1}_{lcu}\ket{0}_{d}\ket{0}_e\ket{\bm{z}} 
         \rightarrow \ket{1}_{lcu}\ket{0}_{d}\sum_n \ket{\widetilde{\lambda}_n}_e \beta_n \ket{n},
      \end{eqnarray}
      where $\widetilde{\lambda}_n$ is the binary approximation of $\lambda_n$, the eigenvalue of applied $\mathcal{D}$.
      \item Then, controlled rotation is applied and only $\ket{0}_d$ part is kept,     
      \begin{eqnarray}
         \ket{1}_{lcu}\ket{0}_{d}\sum_n \ket{\widetilde{\lambda}_n}_e \beta_n \ket{n}&\xrightarrow{C_{R}}&\ket{1}_{lcu} \sum_n \widetilde{\lambda}_{n} \ket{0}_{d} \ket{\widetilde{\lambda}_n}_e \beta_n \ket{n}
      \end{eqnarray}
      \item Finally, phase estimation is inversed,
      \begin{eqnarray}
         \ket{1}_{lcu}\ket{0}_{d} \sum_n \widetilde{\lambda}_n\ket{\widetilde{\lambda}_n}_e \beta_n \ket{n}\rightarrow \ket{1}_{lcu}\ket{0}_d\ket{0}_e\sum_n \widetilde{\lambda}_n \beta_n \ket{n} 
         \rightarrow \ket{1}_{lcu}\ket{0}_d\ket{0}_e\mathcal{D} \ket{\bm{z}},
      \end{eqnarray}
      which accurately established if $\widetilde{\lambda}_n=\lambda_n$. 
      \end{itemize}  
      \State With post-selection, output $\ket{\bm{z}}\pm \xi \mathcal{D}\ket{\bm{z}}$, which is based on $\ket{0}_d\ket{0}_{e}$, which is either result or input of next iteration.
      \end{algorithmic}
     \end{spacing}
\end{algorithm}

To visualize the algorithm, we present a circuit diagram in Figure~\ref{circuit_quantum_loop}. Besides steps depicted in Table.~\ref{algorithm2}, a key component of this algorithm is the implementation of \(e^{-i\mathcal{D}t}\), which will be achieved by querying the coefficient matrix, as detailed in the next section. This crucial step is also represented by the non-trivial section in Figure~\ref{circuit_quantum_loop}.

\begin{figure}[!h]
   \centerline{
   \begin{tikzpicture}[thick]
   \ctikzset{scale=1.8}
   \tikzstyle{every node}=[font=\normalsize,scale=0.9]
     \tikzstyle{operator} = [draw,shape=rectangle
     ,fill=white,minimum width=1em, minimum height=1em] 
     \tikzstyle{operator2} = [draw,shape=rectangle,fill=white,minimum width=3em, minimum height=9.5em] 
     \tikzstyle{operator22} = [draw,shape=rectangle,fill=white,minimum width=3em, minimum height=9.5em] 
     \tikzstyle{operator3} = [draw,shape=rectangle,fill=white,minimum width=3em, minimum height=1em] 
     \tikzstyle{operator4} = [draw,shape=rectangle,dashed, minimum width=1.5cm, minimum height=1cm] 
     \tikzstyle{operator5} = [draw,shape=rectangle,dashed, minimum width=5.75cm, minimum height=3cm] 
     \tikzstyle{operator6} = [draw=pink,shape=rectangle,dashed, minimum width=5cm, minimum height=4cm] 
     \tikzstyle{phase} = [fill,shape=circle,minimum size=3pt,inner sep=0pt]
     \tikzstyle{surround} = [fill=blue!10,thick,draw=black,rounded corners=2mm]
     \tikzstyle{ellipsis} = [fill,shape=circle,minimum size=2pt,inner sep=0pt]
     \tikzstyle{ellipsis2} = [fill,shape=circle,minimum size=0.5pt,inner sep=0pt]
     \tikzset{meter/.append style={fill=white, draw, inner sep=5, rectangle, font=\vphantom{A}, minimum width=15, 
     path picture={\draw[black] ([shift={(.05,.2)}]path picture bounding box.south west) to[bend left=40] ([shift={(-.05,.2)}]path picture bounding box.south east);\draw[black,-latex] ([shift={(0,.15)}]path picture bounding box.south) -- ([shift={(.15,-.08)}]path picture bounding box.north);}}}
     \node at (0.6,0) (qin){$\ket{\bm{z}}$};
     \node at (0.6,-1) (qin){$\ket{\bm{z}}$};
     \node at (-0.1,-1.5) (qin){$\ket{\bm{z}}$};
     \node at (-0.1,-2) (qin){$\ket{0}_e$};
     \node at (-0.1,-2.5) (qin){$\ket{0}_d$};
     \node at (-0.1,-3) (qin){$\ket{0}_{lcu}$};
     \node at (0.75,0) (q2){};
     \node[] (end3) at (3.25,0) {} edge [|-] (q2);
     \node[ellipsis] (op22) at (1,-0.35) {} ;
     \node[ellipsis] (op22) at (1,-0.5) {} ;
     \node[ellipsis] (op22) at (1,-0.65) {} ;  
     \node at (0.75,-1) (q3) {};
     \node[] (end3) at (3.25,-1) {} edge [|-] (q3);
     \node at (4.75,0) (q2){};
     \node[] (end3) at (7.25,0) {} edge [|-] (q2);
     \node at (4.75,-1) (q2){};
     \node[] (end3) at (7.25,-1) {} edge [|-] (q2);
     \node at (4.6,0) (qin){$\ket{\bm{z}}$};
     \node at (4.6,-1) (qin){$\ket{\bm{z}}$};
     \node at (0,-1.5) (q3) {};
     \node[] (end3) at (8.75,-1.5) {$\mbox{updated }\ket{\bm{z}}$} edge [<-] (q3);
     \node at (0,-2) (q4) {};
     \node[] (end3) at (8,-2) {} edge [-] (q4);
     \node at (0,-2.5) (q5) {};
     \node[meter] (end3) at (8.5,-2.5) {} edge [-] (q5);
     \node at (0,-3) (q6) {};
     \node[meter] (end3) at (8.5,-3) {} edge [-] (q6);
     \node[ellipsis] (op22) at (5,-0.35) {} ;
     \node[ellipsis] (op22) at (5,-0.5) {} ;
     \node[ellipsis] (op22) at (5,-0.65) {} ;  
     \node at (1,0) (qin){$/$};
     \node at (1,-1) (qin){$/$};
     \node at (0.25,-1.5) (qin){$/$};
     \node at (0.25,-2) (qin){$/$};
     \node at (5,0) (qin){$/$};
     \node at (5,-1) (qin){$/$};
     \node[operator] (op22) at (0.5,-3) {$R_y(\xi)$} ;
     \node[operator] (op22) at (8,-3) {$R_y(\xi)$} ;
     \node[operator] (op22) at (2,-2) {$H$} ;
     \node[ellipsis] (cc1) at (2.5,-2) {};
     \node[operator2] (op22) at (2.5,-0.75) {$e^{-i\mathcal{M}t/m}$} edge [-] (cc1);
     \node[operator] (op22) at (3,-2) {$U_F$} ;
     \node[ellipsis] (cc2) at (4.5,-2) {};
     \node[operator] (op22) at (4.5,-2.5) {$R_y(\lambda)$} edge [-] (cc2);
     \node[ellipsis] (cc2) at (4.5,-3) {} edge [-] (op22);
     \node[operator] (op22) at (6,-2) {$U_F^{-1}$} ;
     \node[ellipsis] (cc3) at (6.5,-2) {};
     \node[operator2] (op22) at (6.5,-0.75) {$e^{+i\mathcal{M}t/m}$} edge [-] (cc3);;
     \node[operator] (op22) at (7,-2) {$H$} ;
     \draw [cyan, decorate, decoration = {brace}] (1.75,-2.8)--(1.75,-1.3);
     \draw [cyan, decorate, decoration = {brace}] (7.25,-1.3)--(7.25,-2.8);
     \draw [red, decorate, decoration = {brace}] (0.45,-1.2)--(0.45,0.2);
     \draw [red, decorate, decoration = {brace}] (3.3,0.2)--(3.3,-1.2);
     \draw [red, decorate, decoration = {brace}] (4.45,-1.2)--(4.45,0.2);
     \draw [red, decorate, decoration = {brace}] (7.3,0.2)--(7.3,-1.2);
     \node at (7.75,-2.7) (qtex){\textcolor{green}{$QMI$}};
     \node at (3.75,-0.6) (qtex){\textcolor{red}{$\times m $}};
     \node at (7.75,-0.6) (qtex){\textcolor{red}{$\times m$}};
     \node at (3.75,-1.2) (qtex){\textcolor{green}{$e^{-i\mathcal{D} t}$}};
     \node at (7.75,-1.2) (qtex){\textcolor{green}{$e^{+i\mathcal{D} t}$}};
   \end{tikzpicture} } 
   \caption{Circuit diagram to process Eq.~\eqref{iteration_eq}. Component in blue bracket is for Quantum Matrix Inverse and Components in red bracket is for simulating $e^{-i\mathcal{D}t}$.}   
\label{circuit_quantum_loop} 
\end{figure}
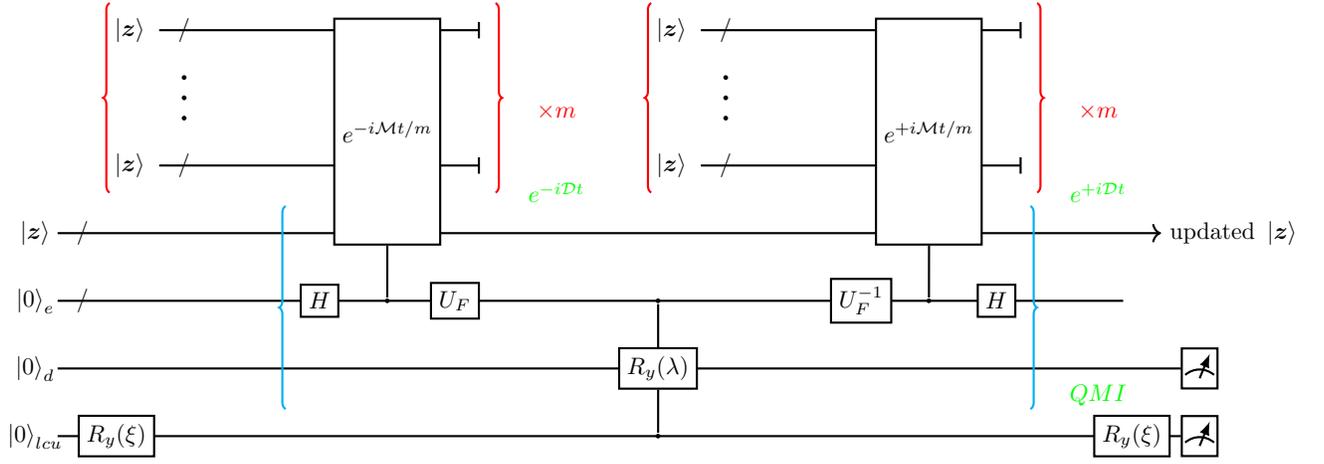

\textit{Implementing gradient operator --- }
The overall framework of the quantum gradient descent algorithm is presented in Table~\ref{algorithm2}, where the implementation of the effective gradient \(\mathcal{D} = Tr_{p-1}\left[\mathbb{I} \otimes \rho_{\bm{x}}^{\otimes p-1} \mathcal{M_D}\right]\) is the central component. In this section, we will provide further details on its implementation.

First, we present all temporary states in Table \ref{algorithm2}. 
Usually, $\mathcal{D}$ is hermitian but not unitary, therefore can not be directly applied in circuit models. But if the evolution $e^{-i\mathcal{D}t}$ can be efficiently realized, we can effectively implement $\mathcal{D}\ket{\psi}$ with the help of an ancillary quantum register in a `$HHL$-like' process as 
\begin{eqnarray}
   \ket{0} \otimes \ket{0}\otimes \ket{\psi}  &&\xrightarrow{H} \ket{0} \otimes \sum_{j=0}^{N-1} \ket{j} \otimes \ket{\psi} \nonumber \\
    && \xrightarrow{C-e^{-i\mathcal{D}t}} \ket{0} \otimes \sum_{j=0}^{N-1} \ket{j}\otimes e^{-i\frac{2\pi}{N}\mathcal{D} j } \ket{\psi}= \ket{0} \otimes \sum_k\sum_{j=0}^{N-1} \beta_k \ket{j}\otimes e^{-i\frac{2\pi}{N}\lambda_k j } \ket{k} \nonumber \\
    && \xrightarrow{QPE}   \ket{0}\otimes \sum_k   \beta_k \ket{\lambda_k} \otimes \ket{k} \nonumber \\
    &&\xrightarrow{C-rotation} \sum_k   \beta_k(\frac{\lambda_k}{C}\ket{0}+  \sqrt{1-|\frac{\lambda_k}{C}|^2}\ket{1}) \otimes  \ket{\lambda_k} \otimes \ket{k} \nonumber \\
    &&\xrightarrow{QPE^{-1}} \sum_k  \sum_{j=0}^{N-1}  \beta_k( \frac{\lambda_k}{C}\ket{0}+\sqrt{1-|\frac{\lambda_k}{C}|^2}\ket{1}) \otimes e^{-i\frac{2\pi}{N} \lambda_k j} \ket{j} \otimes \ket{k} \nonumber \\
    &&=\sum_k  \sum_{j=0}^{N-1}  \beta_k(\frac{\lambda_k}{C}\ket{0}+ \sqrt{1-|\frac{\lambda_k}{C}|^2}\ket{1}) \otimes  \ket{j} \otimes e^{-i\frac{2\pi}{N}\mathcal{D} j } \ket{k} \nonumber \\
    && \xrightarrow{C-e^{i\mathcal{D}t}} \sum_k  \sum_{j=0}^{N-1}  \beta_k(\frac{\lambda_k}{C}\ket{0}+ \sqrt{1-|\frac{\lambda_k}{C}|^2}\ket{1}) \otimes  \ket{j} \otimes \ket{k} \nonumber \\
    && \xrightarrow {H} \sum_k   \beta_k( \frac{\lambda_k}{C}\ket{0}+ \sqrt{1-|\frac{\lambda_k}{C}|^2}\ket{1}) \otimes  \ket{0} \otimes \ket{k} \nonumber \\
    &&\xrightarrow{P=\ket{0}\bra{0}} \sum_k  \beta_k \frac{\lambda_k}{C}\ket{0} \otimes  \ket{0} \otimes \ket{k} \propto \ket{0}\otimes \ket{0} \otimes \mathcal{D}\ket{\psi},
\end{eqnarray}

And in a further step, the evolution $e^{-i\mathcal{D}t}$ can be approximately constructed with another evolution $e^{-i\mathcal{M_D}t}$ and the help of $m(p-1)$ copies of state $\rho_{\bm{x}}=\ket{\bm{x}}\bra{\bm{x}}$, within accuracy $\mathcal{O}(\frac{ t^2p^2 ||\mathcal{F}||_{max}^2}{m})$, in a `quantum principal component analysis(QPCA)' way as 

\begin{align}\label{long_time_evolution}
& \underbrace{Tr_{p-1}[e^{-i \mathcal{M_D} \frac{t}{m}} Tr_{p-1}[e^{-i \mathcal{M_D} \frac{t}{m}} \cdots
    Tr_{p-1}[}_{m-trace} e^{-i \mathcal{M_D} \frac{t}{m}} \rho ^{\otimes p} e^{i \mathcal{M_D} \frac{t}{m}}]\cdots \rho ^{\otimes p-1} e^{i \mathcal{M_D}\frac{t}{m}} ] \rho ^{\otimes p-1} e^{i \mathcal{M}\frac{t}{m}} ]\nonumber \\
&=(e^{-i \mathcal{D} \frac{t}{m}})^m \rho (e^{i \mathcal{D} \frac{t}{m}})^m + \mathcal{O}(m||\mathcal{D}||_{max}^2 \frac{t^2}{m^2}) \nonumber \\
&= e^{-i \mathcal{D}t} \rho e^{i \mathcal{D}t}+ \mathcal{O}(\frac{ t^2p^2 ||\mathcal{F}||_{max}^2}{m}),
\end{align}
since we have $\mathcal{D}=Tr_{p-1}[\mathbb{I}\otimes \rho_{\bm{x}}^{\otimes p-1} \mathcal{M_D}]$ and the evolution  $e^{-i\mathcal{M_D}\frac{t}{m}}$ can be approximated by Trotter expansion
\begin{eqnarray}\label{trotter_expansion}
e^{-i\mathcal{M_D} \frac{t}{m}}=\prod_{k=1}^p \mathcal{P}_k e^{-i \mathcal{F} \frac{t}{m}} \mathcal{P}_k+ \mathcal{O}(\frac{t^2 p^2||\mathcal{F}||_{max}^2}{m^2}),
\end{eqnarray}
which is shown in Figure \ref{realizing_d}.

\begin{figure}[!h]
   \centerline{
   \begin{tikzpicture}[thick]
   \ctikzset{scale=1.4}
   \tikzstyle{every node}=[font=\normalsize,scale=0.7]
     \tikzstyle{operator} = [draw,shape=rectangle
     ,fill=white,minimum width=1em, minimum height=1em] 
     \tikzstyle{operator2} = [draw,shape=rectangle,fill=white,minimum width=3em, minimum height=9.5em] 
     \tikzstyle{operator22} = [draw,shape=rectangle,fill=white,minimum width=3em, minimum height=9.5em] 
     \tikzstyle{operator3} = [draw,shape=rectangle,fill=white,minimum width=3em, minimum height=1em] 
     \tikzstyle{operator4} = [draw,shape=rectangle,dashed, minimum width=1.5cm, minimum height=1cm] 
     \tikzstyle{operator5} = [draw,shape=rectangle,dashed, minimum width=5.75cm, minimum height=3cm] 
     \tikzstyle{operator6} = [draw=pink,shape=rectangle,dashed, minimum width=5cm, minimum height=4cm] 
     \tikzstyle{phase} = [fill,shape=circle,minimum size=3pt,inner sep=0pt]
     \tikzstyle{surround} = [fill=blue!10,thick,draw=black,rounded corners=2mm]
     \tikzstyle{ellipsis} = [fill,shape=circle,minimum size=2pt,inner sep=0pt]
     \tikzstyle{ellipsis2} = [fill,shape=circle,minimum size=0.5pt,inner sep=0pt]
     \tikzset{meter/.append style={fill=white, draw, inner sep=5, rectangle, font=\vphantom{A}, minimum width=15, 
     path picture={\draw[black] ([shift={(.05,.2)}]path picture bounding box.south west) to[bend left=40] ([shift={(-.05,.2)}]path picture bounding box.south east);\draw[black,-latex] ([shift={(0,.15)}]path picture bounding box.south) -- ([shift={(.15,-.08)}]path picture bounding box.north);}}}
     \node at (0.6,0) (qin){$\ket{\bm{z}}$};
     \node at (0.6,-1) (qin){$\ket{\bm{z}}$};
     \node at (0.1,-1.5) (qin){$\ket{\bm{z}}$};
     \node at (0.75,0) (q2){};
     \node[] (end3) at (3.25,0) {} edge [|-] (q2);
     \node[ellipsis] (op22) at (1,-0.35) {} ;
     \node[ellipsis] (op22) at (1,-0.5) {} ;
     \node[ellipsis] (op22) at (1,-0.65) {} ;  
     \node at (0.75,-1) (q3) {};
     \node[] (end3) at (3.25,-1) {} edge [|-] (q3);
     \node at (0.3,-1.5) (q3) {};
     \node[] (end3) at (3.25,-1.5) {} edge [-] (q3);
     \node at (1,0) (qin){$/$};
     \node at (1,-1) (qin){$/$};
     \node at (0.5,-1.5) (qin){$/$};
     \node[operator2] (op22) at (2.5,-0.75) {$e^{-i\mathcal{M}t/m}$};
     \draw [red, decorate, decoration = {brace}] (0.45,-1.2)--(0.45,0.2);
     \draw [red, decorate, decoration = {brace}] (3.3,0.2)--(3.3,-1.2);
     \node at (4,-0.75) (eq1) {$\Leftarrow$};
     \node at (4.6,0) (qin){$\ket{\bm{z}}$};
     \node at (4.6,-1) (qin){$\ket{\bm{z}}$};
     \node at (4.1,-1.5) (qin){$\ket{\bm{z}}$};
     \node at (4.75,0) (q2){};
     \node[] (end3) at (9.,0) {} edge [-] (q2);
     \node at (9.2,0) (q2){};
     \node[] (end3) at (11.75,0) {} edge [|-] (q2);
     \node[ellipsis] (op22) at (5,-0.35) {} ;
     \node[ellipsis] (op22) at (5,-0.5) {} ;
     \node[ellipsis] (op22) at (5,-0.65) {} ;  
     \node at (4.75,-1) (q3) {};
     \node[] (end3) at (9.,-1) {} edge [-] (q3);
     \node at (9.2,-1) (q3){};
     \node[] (end3) at (11.75,-1) {} edge [|-] (q3);
     \node at (4.3,-1.5) (q3) {};
     \node[] (end3) at (9.,-1.5) {} edge [-] (q3);
     \node at (9.2,-1.5) (q3){};
     \node[] (end3) at (11.75,-1.5) {} edge [-] (q3);
     \node at (5,0) (qin){$/$};
     \node at (5,-1) (qin){$/$};
     \node at (4.5,-1.5) (qin){$/$};
     \node[operator2] (op22) at (6.5,-0.75) {$e^{-i\mathcal{F}t/m}$};
     \node[operator2] (op22) at (8.,-0.75) {$e^{-i\mathcal{F}t/m}$};
     \node[operator2] (op22) at (10.5,-0.75) {$e^{-i\mathcal{F}t/m}$};
      \node at (7.25,-1) (cr121) {$\times$};
      \node at (7.25,-1.5) (cr211) {$\times$};
      \node[ellipsis2] at (7.25,-1) (ep121) {};
      \node[ellipsis2] at (7.25,-1.5) (ep211) {};
      \draw[-] (ep121) -- (ep211);
      \node at (8.75,-1) (cr121) {$\times$};
      \node at (8.75,-1.5) (cr211) {$\times$};
      \node[ellipsis2] at (8.75,-1) (ep121) {};
      \node[ellipsis2] at (8.75,-1.5) (ep211) {};
      \draw[-] (ep121) -- (ep211);
      \node at (9.75,-0) (cr121) {$\times$};
      \node at (9.75,-1.5) (cr211) {$\times$};
      \node[ellipsis2] at (9.75,-0) (ep121) {};
      \node[ellipsis2] at (9.75,-1.5) (ep211) {};
      \draw[-] (ep121) -- (ep211);
      \node at (11.25,-0) (cr121) {$\times$};
      \node at (11.25,-1.5) (cr211) {$\times$};
      \node[ellipsis2] at (11.25,-0) (ep121) {};
      \node[ellipsis2] at (11.25,-1.5) (ep211) {};
      \draw[-] (ep121) -- (ep211);
      \node[ellipsis] at (9,-0.75) (ep211) {};
      \node[ellipsis] at (9.1,-0.75) (ep211) {};
      \node[ellipsis] at (9.2,-0.75) (ep211) {};
     \draw [red, decorate, decoration = {brace}] (4.45,-1.2)--(4.45,0.2);
     \draw [red, decorate, decoration = {brace}] (11.75,0.2)--(11.75,-1.2);
   \end{tikzpicture} } 
   \caption{Schematic of implementation of $e^{-i \mathcal{M} \frac{t}{m}}$. The right hand side shows the series performance of $\{\mathcal{P}_k\}$ with $k$ from $1$ to $p$. The "$/$" means multi qubits and "$|$" means dropped of registers in the last of the circuits.}   
\label{realizing_d} 
\end{figure}
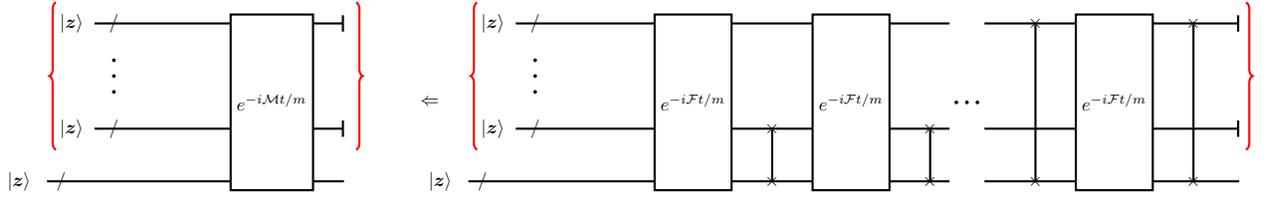

Lastly, as known in general quantum singling processing(QSP) Hamiltonian simulation\cite{low2017optimal}, $e^{-i\mathcal{F}\frac{t}{m}}$ can be simulated within accuracy $\epsilon_h$ at the cost of $\mathcal{O}( sp||\mathcal{F}||_{max} \frac{t}{m} + \frac{\log{1/\epsilon_h}}{\log{\log{1/\epsilon_h}}})
$ queries of $U_1$, $U_2$ and $\mathcal{O}\big(( sp||\mathcal{F}||_{max}\frac{t}{m} + \frac{\log{1/\epsilon_h}}{\log{\log{1/\epsilon_h}}})(d^p +q poly(\log{q}))\big)$ extra basic quantum gates, where $s$ denotes the sparsity of $\mathcal{F}$, $q$ is the bit-accuracy of elements of $\mathcal{F}$, and the two Oracles,
\begin{eqnarray}\label{oracles}
    U_1 \ket{j,k}\ket{0}&&=\ket{j,k}\ket{\mathcal{F}_{j,k}}, \nonumber \\
    U_2 \ket{j,l}&&=\ket{j,k_\mathcal{F}(j,l)}.
\end{eqnarray}

To sum up, we implement the effective gradient $\mathcal{D}$ in the following logical flow,

\begin{eqnarray}
   U_1,U_2 \xrightarrow{QSP}e^{-i\mathcal{F}\frac{t}{m}} \xrightarrow {Trotter} e^{-i\mathcal{M_D}\frac{t}{m}} \xrightarrow{QPCA} e^{-i\mathcal{D}t} \xrightarrow{HHL-like} \mathcal{D}.
\end{eqnarray}

More details have been described in related work in \cite{2019-Rebentrost-Qgradient,gao2021quantum}.

\subsection{Quantum Gradient algorithm with $\mathcal{F}$ Efficiently Decomposed}

If $\mathcal{F}$, the coefficient matrix can be decomposed efficiently, there is alternative method.   
We emphasize one case with the strengthened condition, where $\mathcal{F}$ can be represented by a tensor product of Hermitian operator, which is geometry k-local. Specifically, 
\begin{eqnarray}\label{ee_decomposition}
    \mathcal{F}&=& \otimes_{\alpha=1}^p \otimes_{j=1}^m P_{\alpha,j},   \nonumber \\
    \mathcal{D}&=&\sum_{\alpha=1}^{p} \prod_{\beta\neq\alpha} \bra{\bm{z}} \otimes_{j=1}^m P_{\beta,j} \ket{\bm{z}} \otimes_{j=1}^m P_{\alpha,j},
\end{eqnarray}
where $d = \prod_{j=1}^m \mbox{dimension} (P_{\alpha,j})$ and $\mbox{dimension} (P_{\alpha,j}) \le 2^k$.
Under the condition that $\xi$ is small, we can interpret $I- \xi \mathcal{D}$, which is equivalent to realize Eq.~\eqref{iteration_eq}, as,
\begin{eqnarray}\label{qite}
   e^{-\xi \mathcal{D}} \sim \prod_{\alpha=1}^p  \otimes_{j}^m e^{-\xi \prod_{\beta\neq\alpha} \bra{\bm{z}} \otimes_{j=1}^m P_{\beta,j} \ket{\bm{z}} \otimes_{j=1}^m P_{\alpha,j} }.
\end{eqnarray} 
If $\ket{\bm{z}}$ is chosen as a tensor product state, $ \prod_{\beta\neq\alpha}\bra{\bm{z}} \otimes_{j=1}^m P_{\beta,j} \ket{\bm{z}} \equiv b_{\beta}$ can be efficiently evaluated with $\mathcal{O}(2^L)$ times calculations, where $L\ge k$ is a correlation length.

Here, we introduce two method to simulate $e^{-\xi \mathcal{D}}$.

The first is the one proposed in Ref.\cite{motta2020determining}. 
We set $e^{-i\xi D}= \prod_{\alpha} e^{-i\xi D_{\alpha}}$ and $D_{\alpha} =\sum_{k} a_{\alpha, k} \sigma_{\alpha,k}$.
By approximating $e^{-i\xi D_{\alpha}} \ket{x} = \frac{1}{c} \otimes_{j}^m e^{-\xi b_{\beta}  P_{\alpha,j} }\ket{x}$, a system of linear equation formulates as 
\begin{eqnarray}\label{linear_equation}
  \sum_{k} a_{\alpha, k} \bra{\bm{z}} \sigma^{\dagger}_{\alpha,k'} \sigma_{\alpha,k} \ket{\bm{z}} = \frac{-ib_{\beta}}{c} \bra{\bm{z}} \sigma^{\dagger}_{\alpha,k'} \otimes_{j}^m  P_{\alpha,j} \ket{\bm{z}},
\end{eqnarray}
where $c=\sqrt{\bra{\bm{z}} e^{-2\xi\mathcal{D}}\ket{\bm{z}}}$ is the re-normalization factor.
For the initial step($e^{-i\xi D_{1}}$), one can choose a tensor product state, and coefficients in Eq.~\eqref{linear_equation} can be efficiently obtained. By solving this system, $e^{-i\xi D_{1}}$ is easy to determine. 
For $e^{-i\xi D_{\alpha}}$($\alpha>1$), as $\prod_{\beta=1}^{\alpha-1} e^{-i\xi D_{\beta}}$ has been applied in advance, $\ket{\psi}$ will be non-local and the correlation length will grow with $\alpha$. 
In Ref~\cite{motta2020determining}, one way is to simulate with a larger correlation length($L\ge k$). However, the correlation length will get quickly larger as with $\alpha$ and iterative update with gradient.The other way is an approximate method, where the accuracy depends on the real correlation length and the truncated simulated length.

The second is one with LCU. 
If $P_{\alpha,j}$ in Eq.~\eqref{ee_decomposition} is a $d$-dimension operator that can be presented by Pauli matrix. The method of LCU can be employed to realize the target process with a specified probability.
The gradient operator can be rewritten as,  
\begin{eqnarray}\label{dqc1}
    \mathcal{D}=\sum_{\alpha=1}^{p} \sum_{i=1}^{m} \frac{\prod_{j} \bra{\bm{z}}P_{\alpha,j}\ket{\bm{z}}}{\bra{\bm{z}}F_{\alpha,i}\ket{\bm{z}}} F_{\alpha,i}=\sum_{\alpha=1}^p \sum_{i=1}^m c_{\alpha,i} F_{\alpha,i},
\end{eqnarray} 
where $c_{\alpha,i}$ are the parameters to be determined in advance. In additional experiments, $\mathcal{O}(kp)$ times measurements on $P_{\alpha,j}$ are required for obtaining $c_{\alpha,i}$. The complexity is similar to that of variant quantum algorithm, where each decomposition of $\mathcal{F}$ are required to be measured for values of cost function.
As with realizing $I-\xi \mathcal{D}$, two ancillary registers are required.
One is a single qubit to implement $I-\xi \mathcal{D}$ and the second is $log(kp)$-qubit to implement $\mathcal{D}$. 
Circuit in Fig.~\ref{circuit_lcu} shows the method, where the principal register output $\ket{\bm{z}}$ if post-selections are conducted.
This is with a success possibility at $\mathcal{O}(1/(kp)^2)$ to observe the $\ket{0}_e\ket{0}_d$ subsystem.
More details have been described in related work in \cite{cheng2024polynomial}.

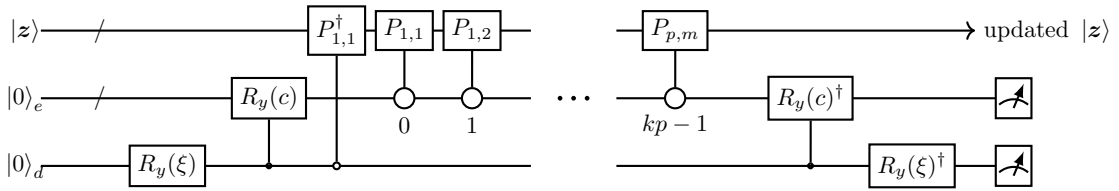
\begin{figure}[!h]
   \centerline{
   \begin{tikzpicture}[thick]
   \ctikzset{scale=1.8}
   \tikzstyle{every node}=[font=\normalsize,scale=0.9]
     \tikzstyle{operator} = [draw,shape=rectangle
     ,fill=white,minimum width=1em, minimum height=1em] 
     \tikzstyle{operator2} = [draw,shape=rectangle,fill=white,minimum width=3em, minimum height=9.5em] 
     \tikzstyle{operator22} = [draw,shape=rectangle,fill=white,minimum width=3em, minimum height=9.5em] 
     \tikzstyle{operator3} = [draw,shape=rectangle,fill=white,minimum width=3em, minimum height=1em] 
     \tikzstyle{operator4} = [draw,shape=rectangle,dashed, minimum width=1.5cm, minimum height=1cm] 
     \tikzstyle{operator5} = [draw,shape=rectangle,dashed, minimum width=5.75cm, minimum height=3cm] 
     \tikzstyle{operator6} = [draw=pink,shape=rectangle,dashed, minimum width=5cm, minimum height=4cm] 
     \tikzstyle{phase} = [fill,shape=circle,minimum size=3pt,inner sep=0pt]
     \tikzstyle{surround} = [fill=blue!10,thick,draw=black,rounded corners=2mm]
     \tikzstyle{ellipsis} = [fill,shape=circle,minimum size=2pt,inner sep=0pt]
     \tikzstyle{ellipsis1} = [draw,shape=circle,minimum size=2pt,inner sep=0pt]
     \tikzstyle{ellipsis3} = [draw,shape=circle,fill=white, minimum size=2pt,inner sep=1pt]
     \tikzstyle{ellipsis4} = [draw,shape=circle,fill=white, minimum size=4pt,inner sep=3pt]
     \tikzstyle{ellipsis2} = [draw,shape=circle,minimum size=0.5pt,inner sep=0pt]
     \tikzset{meter/.append style={fill=white, draw, inner sep=5, rectangle, font=\vphantom{A}, minimum width=15, 
     path picture={\draw[black] ([shift={(.05,.2)}]path picture bounding box.south west) to[bend left=40] ([shift={(-.05,.2)}]path picture bounding box.south east);\draw[black,-latex] ([shift={(0,.15)}]path picture bounding box.south) -- ([shift={(.15,-.08)}]path picture bounding box.north);}}}
     \node at (0.2,-1.5) (qin){$\ket{\bm{z}}$};
     \node at (0.2,-2) (qin){$\ket{0}_e$};
     \node at (0.2,-2.5) (qin){$\ket{0}_d$};
     \node at (0.25,-1.5) (q3) {};
     \node[] (end3) at (4,-1.5) {} edge [-] (q3);
     \node at (4.5,-1.5) (q3) {};
     \node[] (end3) at (7.75,-1.5) {$\mbox{updated}~\ket{\bm{z}}$} edge [<-] (q3);
     \node at (0.25,-2) (q4) {};
     \node[] (end3) at (4,-2) {} edge [-] (q4);
     \node at (4.5,-2) (q4) {};
     \node[meter] (end3) at (7.5,-2) {} edge [-] (q4);
     \node at (0.25,-2.5) (q5) {};
     \node[] (end3) at (4,-2.5) {} edge [-] (q5);
     \node at (4.5,-2.5) (q5) {};
     \node[meter] (end3) at (7.5,-2.5) {} edge [-] (q5);
     \node[ellipsis] (cc1) at (4.15,-2) {};
     \node[ellipsis] (cc1) at (4.25,-2) {};
     \node[ellipsis] (cc1) at (4.35,-2) {};
     \node at (0.75,-1.5) (qin){$/$};
     \node at (0.75,-2) (qin){$/$};
     \node[operator] (op22) at (1.25,-2.5) {$R_y(\xi)$} ;
     \node[ellipsis1] (cc1) at (2,-2.5) {};
     \node[operator] (op22) at (2,-2) {$R_y(c)$} edge [-] (cc1);
     \node[ellipsis3] (cc2) at (2.5,-2.5) {};
     \node[operator] (op22) at (2.5,-1.5) {$P_{1,1}^{\dagger}$} edge [-] (cc2);
     \node[ellipsis4] (cc2) at (3,-2) {};
     \node[operator] (op22) at (3,-1.5) {$P_{1,1}$} edge [-] (cc2);
     \node (cc2) at (3,-2.2) {$0$};
     \node[ellipsis4] (cc2) at (3.5,-2) {};
     \node[operator] (op22) at (3.5,-1.5) {$P_{1,2}$} edge [-] (cc2);
     \node (cc2) at (3.5,-2.2) {$1$};
     \node[ellipsis4] (cc2) at (5,-2) {};
     \node[operator] (op22) at (5,-1.5) {$P_{p,m}$} edge [-] (cc2);
     \node (cc2) at (5,-2.2) {$kp-1$};
     \node[ellipsis1] (cc1) at (6,-2.5) {};
     \node[operator] (op22) at (6,-2) {$R_y(c)^{\dagger}$} edge [-] (cc1);
     \node[operator] (op22) at (6.75,-2.5) {$R_y(\xi)^{\dagger}$} ;
   \end{tikzpicture} } 
   \caption{Circuit to process Eq.~\eqref{iteration_eq} with method of linear combination of unitaries. }   
\label{circuit_lcu} 
\end{figure}

\section{Quantum circuit synthesis with reinforcement learning}
\label{supp:RL}
For a given initial state $\ket{z}$ and target state $\ket{z'}$ , the RL algorithm learns a policy that generates a PQC (by adding gates sequentially), which transform $\ket{z}$ to approximate $\ket{z'}$. The policy is learned by Proximal Policy optimization (PPO) algorithm \cite{schulman2017proximal} with the built-in MLP functional approximator. The reward is set to 
\begin{equation}
r = \left\{ 
   \begin{array}{ll}
           10 & if \quad f \ge f_t   \\
           -5.0 & if \quad f< f_t \quad and \quad circuit\_len >= max\_circuit\_len  \\
           \max(\frac{f-f_{prev}}{f_t - f_{prev}}, -1.0) - p & elsewise
    \end{array}
    \right.
\end{equation}
where $f = \bra{z'}U({\theta})\ket{z}$ is the fidelity of current step (circuit), $f_{prev}$ is the fidelity from last step of the episode and $f_t$ is a pre-defined fidelity threshold. $p$ is a penalty of adding a gate to encourage shallow circuit which makes the agent always trying to find the shortest possible protocol to prepare the target state.  The fidelity $f$ is obtained by optimizing all parameters $\{\theta\}$ of the  quantum circuit for current step. The optimization is done using the COBYLA algorithm in scipy optim package. The episode terminates if fidelity value is great than $f_t$ or the circuit depth(number of gates) exceeds the maximum allowed circuit depth.  This is to prevent episodes, especially at the beginning of training,  become exceedingly long leading to unfeasible training times. All the quantum circuit simulation is done using the Qulacs python package \cite{suzuki2021qulacs}.

The pseudo-code of reinforcement algorithm for PQC learning is concluded in Table~\ref{pqc_learning}.
\begin{figure}[H]
\floatname{algorithm}{Table}
\begin{algorithm}[H]
\caption{Reinforcement Learning Procedure for PQC}
\label{pqc_learning} 

\begin{algorithmic}[1]
\State \textbf{Input:} Initial state $\ket{z}$,  target state $\ket{z}'$, maximum circuit depth $d_{max}$, reward penalty $p$, fidelity threshold $f_{t}$, convergence threshold $\epsilon$
\State \textbf{Output:} Learned policy for circuit synthesis and corresponding PQC
\State Initialize gym environment env($\ket{z}$, $\ket{z}'$, $d_{max}$, $p$,  $f_{t}$)
\State $\theta_{\text{prev}} \gets \theta$
\State Define actions (quantum gates set ${g_{i}(\theta_i)}$), reward function and termination condition.
\Repeat
    \State Learn PPO($\gamma$, $n_{epochs}$, $clip\_range$, $learning\_rate$)
\Until{maximum learning steps}

\State Reset environment
\Repeat
     \State Predict action for the current state: $a$ = PPO$_{opt}(state)$
     \State Take a step: $state$, $reward$, $done$, $info$ = env.step($a$)
\Until{$done$}

\State \textbf{Return:} Learned policy $PPO_{opt}$ and parameterized quantum circuit

\end{algorithmic}
\end{algorithm}
\end{figure}

For both examples (Max-cut problem  and polynomial optimization), we choose from the gate set consisting of single qubit Pauli rotation gates and two-bit Pauli rotation gates, namely {$R_X(\theta)$, $R_Y(\theta)$, $R_Z(\theta)$, $R_{XX}(\theta)$, $R_{YY}(\theta)$, $R_{ZZ}(\theta)$} (by definition $R_{X}(\theta) = exp(-\textit{i}\frac{\theta}{2}\textit{X})$, $R_{XX}(\theta) = exp(-\textit{i}\frac{\theta}{2}\textit{X}\times\textit{X}) $ etc.). In the learning process, we set the maximum circuit depth to be 10.  The parameter setting of PPO algorithm is $\gamma = 0.99$, $n_{epochs} = 4$, $clip\_range = 0.2$,  $learning\_rate = 0.0001$.

We tested the algorithm on random 4-qubit state with different quantum gradient step size, resulting in different target state. Figure \ref{test_r1} shows the influence of number of learning episodes on the fidelity of state after learned PQC and the target state with different. We can see that, in general, the RL algorithm can approximate target state better with increased learning episode, but reaching some plateau potentially restricted by the maximum episode length (e.g. circuit depth). 

\begin{figure}[!ht]
   \centering
   \includegraphics[width=0.5\columnwidth]{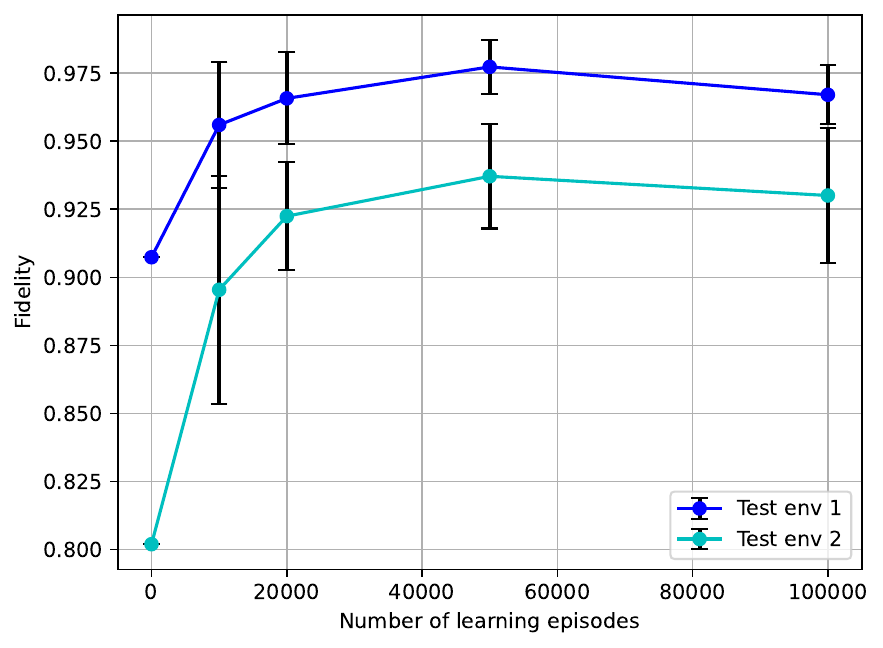}
   \caption{Influence of number of learning episodes on learned policy: two environments are tested, with different target state. Mean and standard deviation of fidelity values are calculated from five independent runs (fidelity between initial state and target states is shown at $x=0$).} 
\label{test_r1} 
\end{figure}




\section{Details on Numerical simulation}
\label{supp:simu}

\textit{Supplementary for Numerical Results in Manuscript ---}
First, we specify our problems for simulation. 
The first is a Max-cut problem. Suppose there are $n$ vortex in the graph, the observable that we aim to maximize is
\[
H_c = \sum_{\langle i,j \rangle} \frac{1 - \sigma_{z_i} \sigma_{z_j}}{2},
\]
where \(\sigma_{z_i}\) and \(\sigma_{z_j}\) represent the Pauli-Z operators acting on qubits \(i\) and \(j\), $\langle i,j \rangle$ exists only there is an edge. 

We are noticed of the results by Edward~\cite{farhi2014quantum}.
The mixer Hamiltonian can be chosen as 
\begin{eqnarray}
    H_b=\sum_{i=1}^n \sigma_{x_i},
\end{eqnarray}
and can be used as QAOA method, where $\sigma_{x_i}$ is the Pauli-X operators acting on qubits \(i\). The results indicate that if the graph is a ring(that is, regular-2 graph), the maximum value of the objective function $\langle H_c \rangle$, which is estimated under state output by a QAOQ circuit is bounded as $n(2p+1)/(2p+2)$, where $p$ is the depth of QAOA circuits.
Therefore, increasing $p$ will produce a better approximate ratio for more general graphs.

Specifically, we consider a graph, which is a random graph with 4 vortex. 
\begin{figure}[!ht]
   \centering
   \includegraphics[width=0.6\columnwidth]{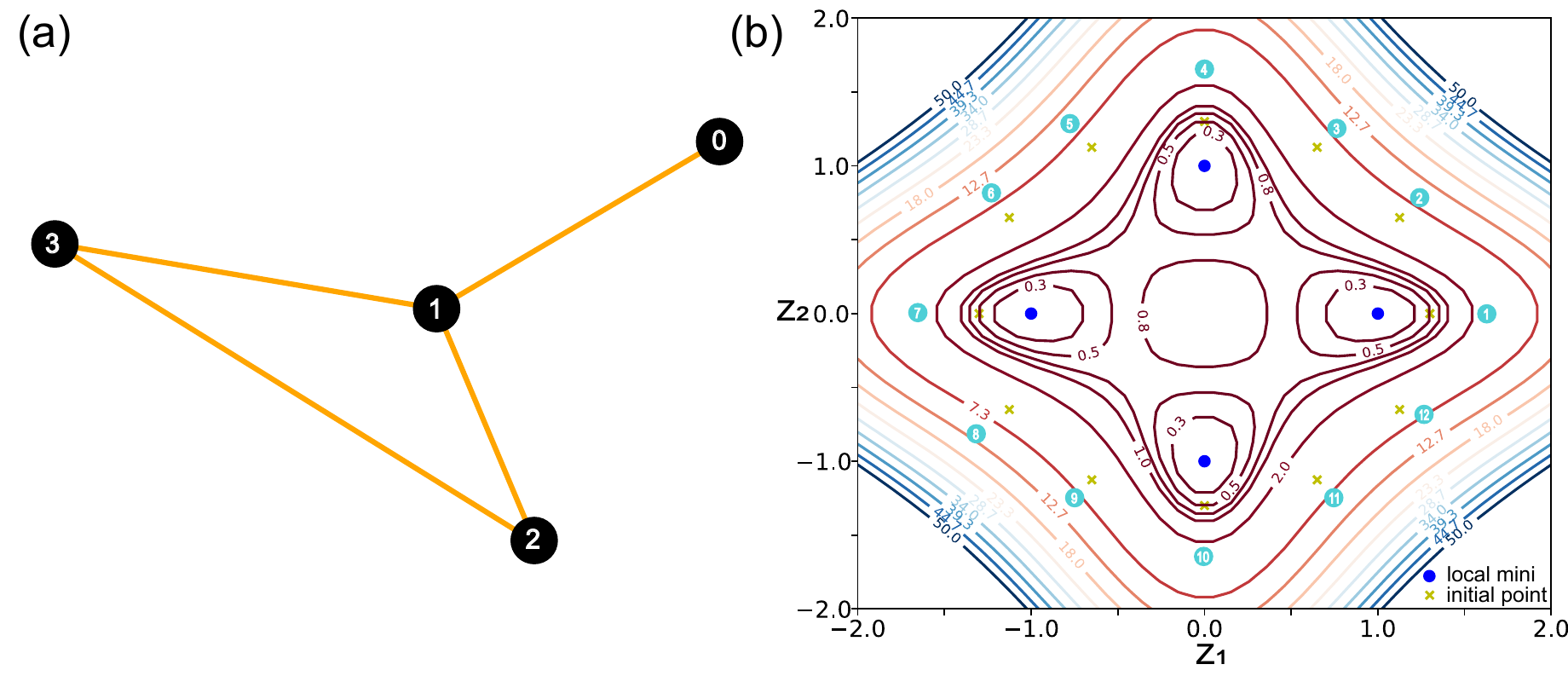}
   \caption{(a) is a random graph with 4 vortex. (b) is the contour plot of a polynomial within $z_3=0$ plane} 
\label{fig_graph} 
\end{figure}
The corresponding Hamiltonian is depicted as 
\[
H_c = \sum_{\langle i,j \rangle} \frac{1 - \sigma_{z_i} \sigma_{z_j}}{2}, 
\]
where $\langle i,j \rangle$ exists only there is an edge in Figure \ref{fig_graph}(a). 

In our simulation, we set $4$ quantum circuit as our initialization.
The first one formulates as a single layer QAOA ansatz, 
\[
e ^{-i H_b \beta}e^{-iH_c \gamma}, 
\]
where $H_b =\sum_{i=0}^4 \sigma_{x_i}$. The second one is a modified single layer QAOA ansatz, 
\[
e ^{-i H_b \beta}R_x(\alpha)e^{-iH_c \gamma},
\]
where $R_x(\alpha)$ is a global rotation on x-axis with $\alpha$.
The third one is inspired by hardware efficient ansatz, which is 
\[
e ^{-i \sigma_{z_0}\sigma_{z_3} \alpha}e ^{-i \sigma_{z_1}\sigma_{z_2} \beta}e ^{-i \sigma_{x_2}\sigma_{x_3} \gamma}e ^{-i \sigma_{x_0}\sigma_{x_1} \eta}. 
\]
The fourth one is 
\[
e ^{-i \sigma_{z_0}\sigma_{z_3} \alpha}e ^{-i \sigma_{z_1}\sigma_{z_2} \beta}R_y(\alpha)e ^{-i \sigma_{x_2}\sigma_{x_3} \gamma}e ^{-i \sigma_{x_0}\sigma_{x_1} \eta},
\]
where $R_y(\alpha)$ is a global rotation on x-axis with $\alpha$.
Then using typical method we can find $4$ corresponding sub-optimal parameter configurations which is not the ideal answer but with gradients vanished. 
Then we use these circuit and their parameter configurations as our initialization and conduct the NOM as depicted in pseudo-codes. Based on this, we can get the results of demonstration in the main text.

The second problem to show in our simulation is to optimize a polynomial. Polynomials are a kind of function that has many applications. For example,  nonlinear equations and linear regression. Here we show a case that is optimizing of a polynomial
\begin{eqnarray}
  f=&&z_{1}^{2}\overline{z_{1}}^{2}  +  z_{2}^{2} \overline{z_{2}}^{2} +z_{3}^{2}\overline{z_{3}}^{2} + 2 (\overline{z_{1}}^{2}z_{2}^{2}+  z_{1}^{2} \overline{z_{2}}^{2} )  \nonumber \\
  && +6z_{1} z_{2} \overline{z_{1}} \overline{z_{2}} - 2 z_{1} z_{3} \overline{z_{1}} \overline{z_{3}}  - 2 z_{2} z_{3} \overline{z_{2}} \overline{z_{3}}  \nonumber \\
  &&- 2 z_{1} \overline{z_{1}}- 2 z_{2} \overline{z_{2}}  + 6 z_{3} \overline{z_{3}}  \nonumber \\
  && +  2 ( z_{3}^{2}+\overline{z_{3}}^{2}) + 1
\end{eqnarray}
which can be expressed as 
\begin{eqnarray}
    f(Z)=Z^{\otimes 2} A Z^{\dagger  \otimes 2},
\end{eqnarray}
where $Z=(1,z_1,z_2,z_3)$ and the coefficient matrix 
\begin{eqnarray}
    A=X^{\otimes 4}+ Y^{\otimes 4}+Z^{\otimes 4}.
\end{eqnarray}

It is observed that the cost function features isolated local minima within the region \((z_1,z_2,z_3) \in [-2, 2]^{\otimes 3}\), and our simulation focuses on optimizing within this region. In Figure~\ref{fig_graph}(b), we present a case where \(z_3 = 0\). We start with $12$ initial points, each positioned on a circle with a radius of $1.3$. To generate these initial points, we use a parameterized circuit to approximate them via a state-to-state method. These circuit with their parameter configurations serve as our initialization for the method. Then we conduct the NOM as depicted in pseudo-codes. Based on this, we can get the results of demonstration in the main text.



\textit{Influence by PQC Approximation Method ---}
We additionally investigate how inaccuracy of PQC approximation method affects the protocol, which is RL method work for. In this simulation, we consider the same demonstration problems and the settings are the same as previous ones.

Figures~\ref{num_r1} and~\ref{num_r2} illustrate the deviation between the ideal scenario—where every updated state from the quantum gradient algorithm is assumed to be perfectly captured by a PQC—and the scenario where the state is approximated by a PQC with added random noise, weighted by the magnitude of disturbance. In these cases, we chose $0.2$ as learning rate in simulation of quantum algorithm.

We examine the influence of disturbance magnitude on each iteration, as shown in sub-figures (a) of both scenarios. Snapshots at three disturbance magnitudes—0.02(0.02), 0.06(0.04), and 0.10(0.06)—are presented in sub-figures (b)–(d) of both figures. The magnitude of disturbance reflects the accuracy of the classical training process in approximating the quantum gradient state with a PQC. The results in Figures \ref{num_r1} and \ref{num_r2} indicate that the protocol remains robust under a certain level of disturbance, but if the approximation quality deteriorates significantly, convergence is disrupted. 

\begin{figure}[!ht]
   \centering
   \includegraphics[width=0.8\columnwidth]{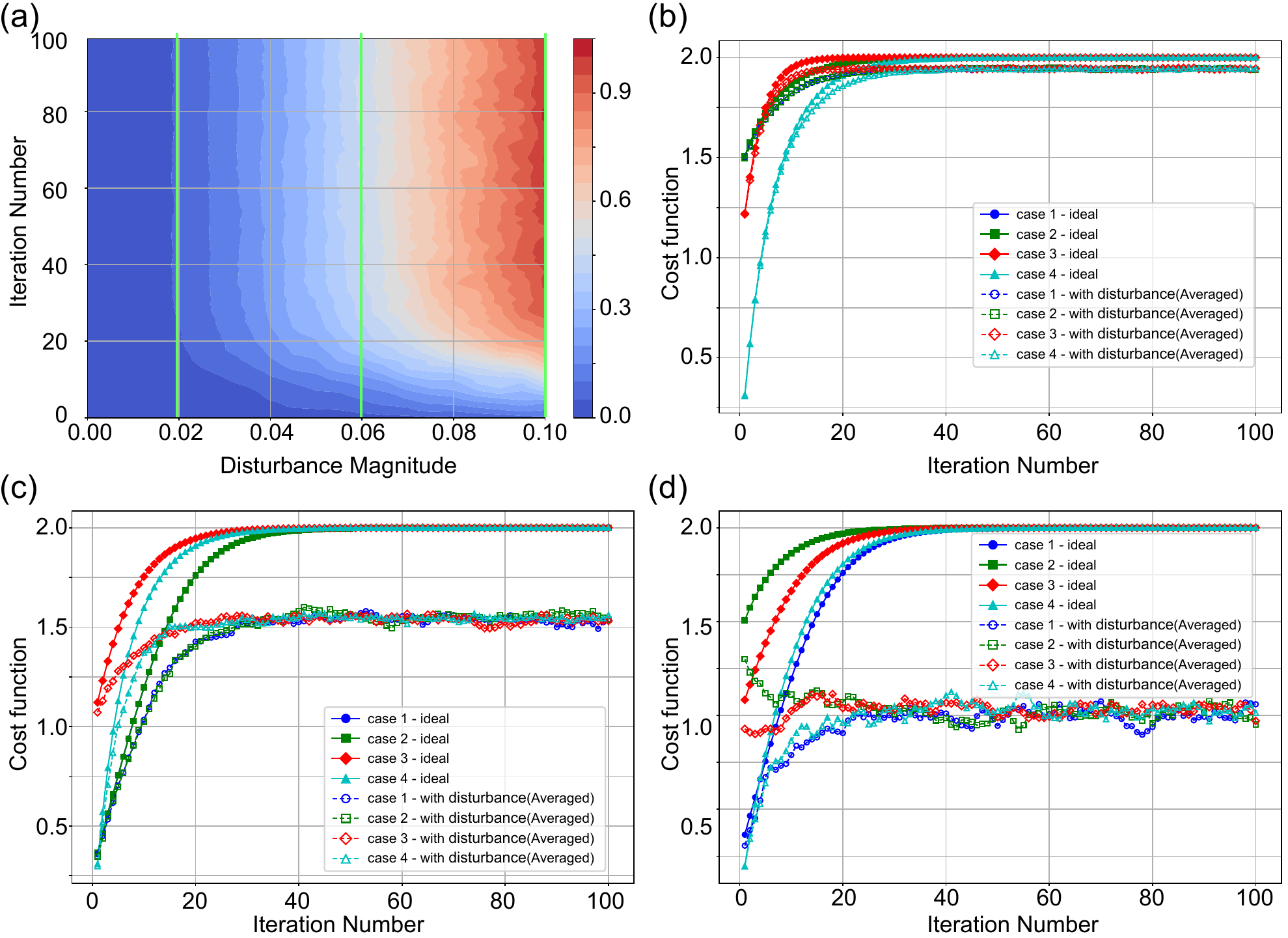}
   \caption{Influence of disturbance magnitude on the MaxCut problem: $4$ initial ansatze are tested. Sub-figure (a) shows the overall effect of varying disturbance magnitudes across all iterations. Sub-figures (b), (c), and (d) provide detailed snapshots at specific disturbance magnitudes of 0.01, 0.03, and 0.05, respectively. The results are averaged over 50 runs to ensure statistical reliability.} 
\label{num_r1} 
\end{figure}

\begin{figure}[!ht]
   \centering
   \includegraphics[width=0.8\columnwidth]{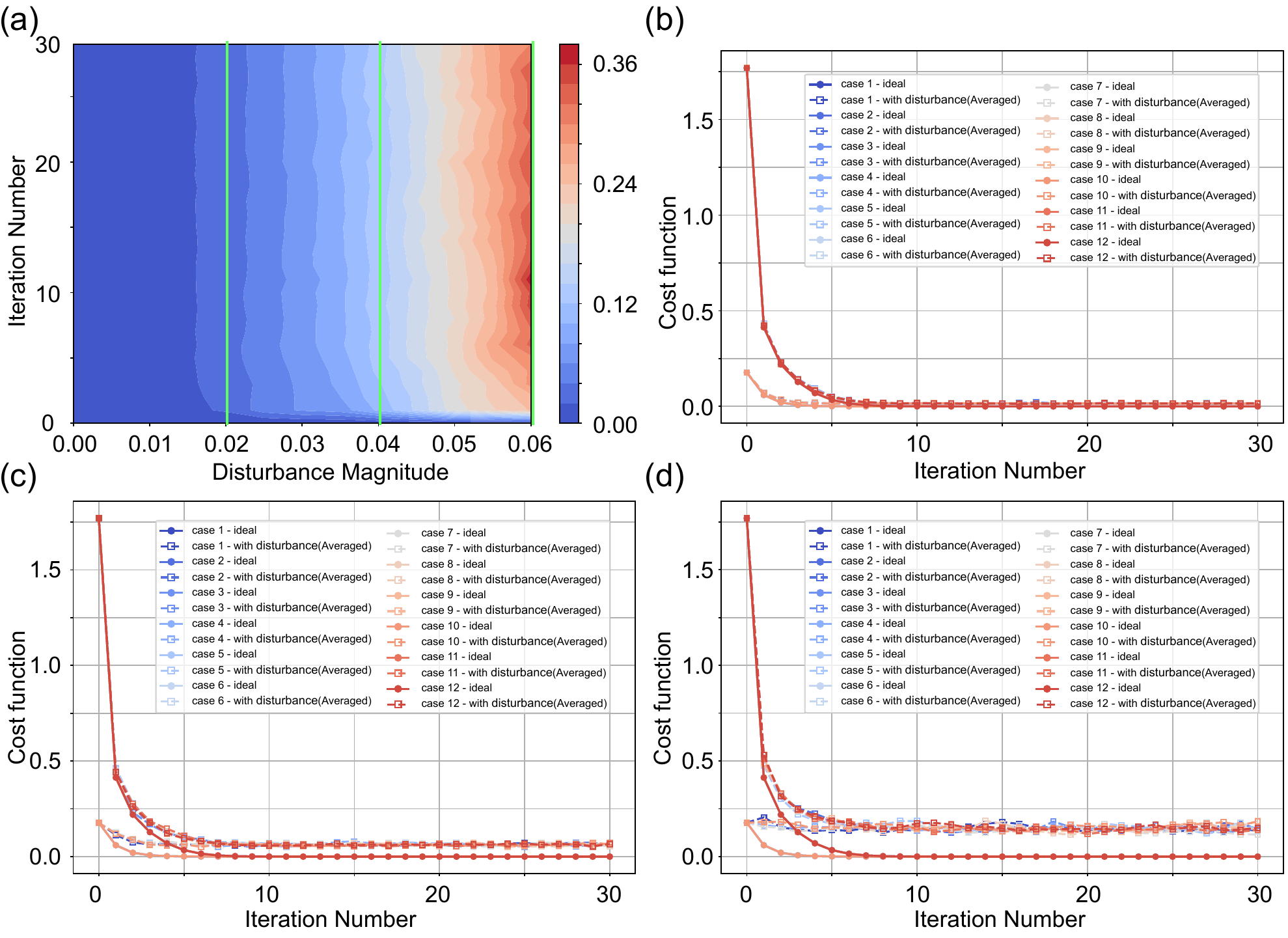}
   \caption{Influence of disturbance magnitude on polynomial optimization: 12 initial points are tested. (a) depicts the overall effect of varying disturbance magnitudes across all iterations.  (b), (c), and (d) present detailed snapshots at specific disturbance magnitudes of 0.01, 0.03, and 0.05, respectively. The results are averaged over 50 runs to ensure robustness.} 
\label{num_r2} 
\end{figure}

\section{Error Accumulation and Analysis}
\label{supp:error}
In this section, we will depicts the fact that the iterative method is insensitive to error of variables  produced either in preparation or readout.

\begin{proposition}
        The iterative gradient algorithm is insensitive to errors from an inaccurate quantum state. 
\end{proposition}

    The inaccurate quantum state is caused by either inaccuracy of devices or optimization methods, which leads a perturbation on output variable. As the quantum gradient algorithm faithfully performs calculation of gradients and the iterative operations. 
    
    We directly analyze its classical behavior. The iterative equations are listed,
    \begin{eqnarray}\label{iteration_eq2}
       \bm{x}'\leftarrow \bm{x}- \xi \nabla_{\bm{x}}f(\bm{x}),
    \end{eqnarray}
    If $\bm{x}$ has a small perturbation $\bm{x}+\Delta$, Eq.~\eqref{iteration_eq2} becomes,
    \begin{eqnarray}\label{eq:d2}
       \bm{x}'\leftarrow && \bm{x}+\Delta - \xi \nabla_{\bm{x}}f(\bm{x}+\Delta)\nonumber \\
       = &&\bm{x}+\Delta - \xi (\nabla_{\bm{x}}f(\bm{x}) + \nabla^2_{\bm{x}}f(\bm{x}) \Delta),
    \end{eqnarray}
    where $\Delta- \xi \nabla^2_{\bm{x}}f(\bm{x}) \Delta$ implies the perturbation on updated point, which is curved back by $\xi \nabla^2_{\bm{x}}f(\bm{x}) \Delta$, denoted as $\Delta'$ (shown in Fig.\ref{grad_robustness}). This mitigation can be realized by adjusting $\xi$. And a special case is in the right-side of the figure with proper $\xi$ and strength of $\Delta$,
    \begin{eqnarray}
        \bm{x}\pm \xi \nabla_{\bm{x}}f(\bm{x}).
    \end{eqnarray}

\begin{figure}[!ht]
   \centering
   \includegraphics[width=0.6\columnwidth]{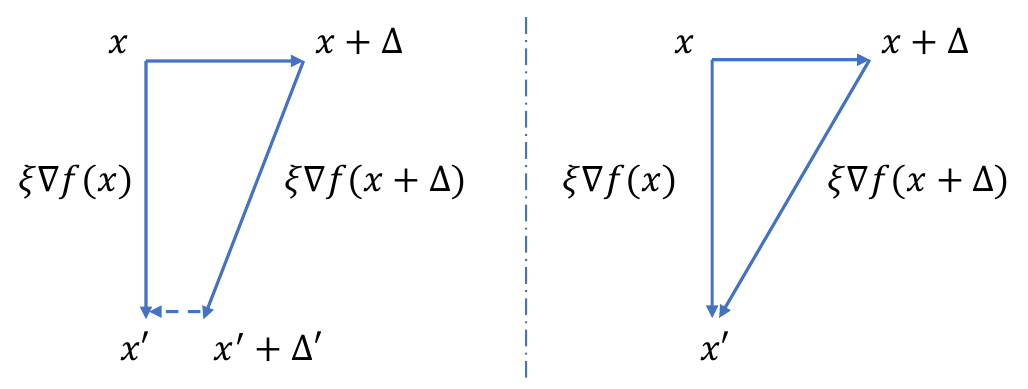}
   \caption{Gradient calculation with errors.} 
\label{grad_robustness} 
\end{figure}

\begin{proposition}
    if quantum gradient operator is inaccurate, the accumulated errors can be mitigated by properly setting learning rate at each iteration. It consists of error from current operator and states affected by previous accumulations.
\end{proposition}

We consider an ideal situation that quantum state can be accurately prepared and measured. However, the gradient operator is inaccurate, which means,  $\Tilde{\nabla}_{\bm{x}}f(\bm{x})=\nabla_{\bm{x}}f(\bm{x})+\bm{e}_{\bm{x}}$. By employing the iterative equation, we have 
\begin{eqnarray}
    \bm{x}'\leftarrow \bm{x}- \xi (\nabla_{\bm{x}}f(\bm{x})+\bm{e}_{\bm{x}})
\end{eqnarray} 
which indicates $\Delta =\xi \bm{e}_{\bm{x}}$. For the next iteration, we use the Eq.~\eqref{eq:d2},
\begin{eqnarray}
    \bm{x}'\leftarrow \bm{x}-\xi\nabla_{\bm{x}}f(\bm{x})+ 
    \xi e_{\bm{x}}-\xi (\nabla^2_{\bm{x}}f(\bm{x})\xi e_{\bm{x}}+\bm{e}_{\bm{x}+\xi e_{\bm{x}}}),
\end{eqnarray}
where the error term $\xi e_{\bm{x}}-\xi(\nabla^2_{\bm{x}}f(\bm{x})\xi e_{\bm{x}}+\bm{e}_{\bm{x}+\xi e_{\bm{x}}})$ is generated and is an adaptive one. 

$\xi \bm{e}_{\bm{x}+\xi e_{\bm{x}}}$ is error from the current gradient operator, and 
$\xi e_{\bm{x}}-\xi(\nabla^2_{\bm{x}}f(\bm{x})\xi e_{\bm{x}})$ is the error in Eq.~\eqref{eq:d2}, which is error from state but is still via imperfection of gradient operator in previous iterations. By choosing proper $\xi$, this error can also be mitigated.

\section{Details on Feasibility}
\textit{Analysis of barren plateaus in classical training part---}
Classical gradient-based methods for training $T(\bm{\alpha})$ may encounter barren plateaus. In this section, we analyze this situation and clarify the mitigation for such barren plateaus.

We focus on $T(\bm{\alpha}) = \prod_{j=1}^{L} T_j(\alpha_j)$, where $T_j(\alpha_j)$ defines a single-layer circuit, and $L$ is finite due to the choice of $\xi$ that limits $\xi \mathcal{D}$. 

The partial derivative with respect to the $k$-th parameter for training $T(\bm{\alpha})$ is given by, 
\begin{equation}\label{sp:bp_gradient1}
    \frac{\partial c_2}{\partial \alpha_k } = i \bra{\bm{z}}T^{\dagger}_{-} \left[ V_k, ~ T^{\dagger}_{+} \ket{\bm{z}'}\bra{\bm{z}'} T_{+} \right] T_{-} \ket{\bm{z}},
\end{equation} 
where $V_k$ denotes the derivative of a single-layer circuit $T_k(\alpha_k)$, $[\cdot]$ represents the commutator, and 
\begin{eqnarray}
    T_{-} = \prod_{j=1}^{k} T_j(\alpha_j), \quad 
    T_{+} = \prod_{j=k+1}^{L} T_j(\alpha_j). \nonumber 
\end{eqnarray}

Let us observe the behavior at the initial moment. We initialize $T(\bm{\alpha}) = T_{+}T_{-}$ to the identity $I$. In this case, 
\begin{eqnarray}\label{sp:bp_gradient2}
\frac{\partial c_2}{\partial \alpha_k } 
 &=&i \bra{\bm{z}} \left( T^{\dagger}_{-}V_kT^{\dagger}_{+} \ket{\bm{z}'}\bra{\bm{z}'} -  \ket{\bm{z}'}\bra{\bm{z}'} T_{+}V_kT_{-} \right)  \ket{\bm{z}} \nonumber \\
 &=&i\bra{\bm{z}} \left[\ket{\bm{z}'}\bra{\bm{z}'}, ~ V_k \right]\ket{\bm{z}}, 
\end{eqnarray} 
where the second line holds due to the initialization $T_{+} = T_{-}^{\dagger} = I$. Thereby, Eq.~\eqref{sp:bp_gradient2} equals zero only when the target state commutes with the observable, which is generally not the case~\cite{grant2019initialization}. 

When $T_{+}$ and $T_{-}$ are not initialized from the identity, we have 
\begin{eqnarray}
\frac{\partial c_2}{\partial \alpha_k } 
 = i \bra{\bm{z}} \left( T^{\dagger}_{-}V_kT^{\dagger}_{+} \ket{\bm{z}'}\bra{\bm{z}'} -  \ket{\bm{z}'}\bra{\bm{z}'} T_{+}V_kT_{-} \right)  \ket{\bm{z}}. \nonumber 
\end{eqnarray} 
Given the precondition that the search region is centered around the identity, meaning 
$d(T(\bm{\alpha}), I)=d(T_{+},T_{-}^{\dagger}) $ is bounded by $ \epsilon_{f}$, then $|d(T_{+}, I)-d(T_{-},I)| < d_f < d(T_{+}, I) + d(T_{-},I)$.
Thus, without normalization, we have
\begin{eqnarray}
    T(\bm{\alpha})\ket{z} = \ket{z} + \epsilon U' \ket{z}, \nonumber \\
    T_{-}\ket{z} = \ket{z} + k_{-}\epsilon U' \ket{z}, \nonumber \\
    T_{+}\ket{z} = \ket{z} + k_{+}\epsilon U' \ket{z},
\end{eqnarray}
where $k_{+}$ and $k_{-}$ are related to the searching region of $T_{+}$ and $T_{-}$, restricted by their depth.
We can then transform the partial derivative into
\begin{eqnarray}
\frac{\partial c_2}{\partial \alpha_k } 
 &=& i \bra{\bm{z}} \left(  V_k \ket{\bm{z}'}\bra{\bm{z}'} -  \ket{\bm{z}'}\bra{\bm{z}'} V_k \right)  \ket{\bm{z}} \nonumber \\
 &+& ik_{+}\epsilon \bra{\bm{z}} \left( V_kU^{\dagger}_{+} \ket{\bm{z}'}\bra{\bm{z}'} -  \ket{\bm{z}'}\bra{\bm{z}'} U_{+}V_k \right)  \ket{\bm{z}} \nonumber \\
 &+& ik_{-}\epsilon \bra{\bm{z}} \left( U^{\dagger}_{-}V_k \ket{\bm{z}'}\bra{\bm{z}'} -  \ket{\bm{z}'}\bra{\bm{z}'} V_kU_{-} \right)  \ket{\bm{z}}. \nonumber \\
\end{eqnarray}
This shows that the major part of the partial derivative is Eq.~\eqref{sp:bp_gradient2}, can be preserved by maintaining a finite search region for $T(\bm{\alpha})$, thereby mitigating the effects of barren plateaus.


\end{widetext}

\end{document}